\documentclass[11pt,a4paper]{article}
\usepackage{amsfonts,amsmath,amssymb,amsthm}
\usepackage{tikz}
\usepackage{caption}
\usepackage{enumerate}
\usepackage{latexsym,amscd}
\usepackage{amsbsy}
\usepackage{CJK}
\usepackage{indentfirst,mathrsfs}
\usepackage{latexsym,amscd}
\usepackage{amsbsy}
\usepackage[noadjust]{cite}
\usepackage{color}
\usepackage{multirow}
\usepackage{makecell}
\usepackage{float}
\usepackage{booktabs}
\usepackage[margin=2.5cm]{geometry}

\usepackage{scalerel,stackengine}
\stackMath
\newcommand\reallywidehat[1]{
\savestack{\tmpbox}{\stretchto{
  \scaleto{
    \scalerel*[\widthof{\ensuremath{#1}}]{\kern-.6pt\bigwedge\kern-.6pt}
    {\rule[-\textheight/2]{1ex}{\textheight}}
  }{\textheight}
  }{0.5ex}}
\stackon[1pt]{#1}{\tmpbox}
}

\numberwithin{equation}{section}
\newtheorem{theo}{Theorem}
\newtheorem{lem}{Lemma}
\newtheorem{coro}{Corollary}

\newtheorem{definition}{Definition}
\newtheorem{example}{Example}
\newtheorem{rem}{Remark}

\title{New constant-dimension subspace codes from parallel cosets of optimal Ferrers diagram rank-metric codes and multilevel inserting constructions }
\author{Gang Wang\textsuperscript{1$^{\ast}$} \and Hong-Yang Yao\textsuperscript{1} \and Fang-Wei Fu\textsuperscript{2}}
\date{\small\textsuperscript{1}College of Science, Civil Aviation University of China, 300300, Tianjin, China. E-mail: gwang06080923@mail.nankai.edu.cn; hyyao0603@163.com.\\
\small\textsuperscript{2}Chern Institute of Mathematics and LPMC, Nankai University, 300071, Tianjin, China. E-mail: fwfu@nankai.edu.cn.
\\$^*$Corresponding author}

\begin{document}

\maketitle

\begin{abstract}
Constant-dimension subspace codes (CDCs), a special class of subspace codes, have attracted significant attention due to their applications in network coding. A fundamental research problem of CDCs is to determine the maximum number of codewords under the given parameters. The paper first proposes the construction of parallel cosets of optimal Ferrers diagram rank-metric codes (FDRMCs) by employing the list of CDCs and inverse list of CDCs. Then a new class of CDCs is obtained by combining the parallel cosets of optimal FDRMCs with parallel linkage construction. Next, we present a novel set of identifying vectors and provide a new construction of CDCs via the multilevel constuction. Finally, the coset construction is inserted into the multilevel construction and three classes of large CDCs are provided, one of which is constructed by using new optimal FDRMCs. Our results establish at least 65 new lower bounds for CDCs with larger sizes than the previously best known codes.
\end{abstract}



\section{Introduction}

Subspace codes, particularly constant-dimension subspace codes (CDCs), have attracted great research interest owing to their extensive applications in random network coding. Following the foundational work by Kötter and Kschischang \cite{1}, subspace codes have demonstrated significant advantages in both error-correction capability and transmission efficiency over operator channels.

Let \( q \) be a prime power, \( \mathbb{F}_q \) be the finite field of order \( q \) and \( \mathbb{F}_q^n \) be an \( n \)-dimensional vector space over \( \mathbb{F}_q \).  
The set of all subspaces of \( \mathbb{F}_q^n \) forms the projective space of order \( n \) over \( \mathbb{F}_q \), denoted by \( \mathcal{P}_q(n) \). 
The projective space \( \mathcal{P}_q(n) \), as metric space, is endowed with the natural measure $d_S(\mathcal{U}, \mathcal{V}) \triangleq \dim(\mathcal{U}+\mathcal{V}) - \dim(\mathcal{U} \cap \mathcal{V})$ for any two distinct subspaces \( \mathcal{U}, \mathcal{V} \in  \mathcal{P}_q(n) \).
For a non-negative integer \( k \leq n \), the collection of all \( k \)-dimensional subspaces of \( \mathbb{F}_q^n \) constitutes the Grassmannian \( \mathcal{G}_q(n, k) \), whose cardinality is given by the Gaussian binomial coefficient

\[
|\mathcal{G}_q(n, k)| = \binom{n}{k}_q \triangleq \prod_{i=0}^{k-1} \frac{q^{n-i} - 1}{q^{k-i} - 1}.
\]

For two subspaces \( \mathcal{U}, \mathcal{V} \in \mathcal{G}_q(n, k) \), the subspace distance is
\begin{align}
d_S(\mathcal{U}, \mathcal{V}) &\triangleq \dim(\mathcal{U}) + \dim(\mathcal{V}) - 2 \dim(\mathcal{U} \cap \mathcal{V}) \nonumber \\
&= 2k - 2 \dim(\mathcal{U} \cap \mathcal{V})\nonumber.
\end{align}
A non-empty subset \( \mathcal{C} \subseteq \mathcal{G}_q(n, k) \) is called an \((n, M, d, k)_q\)-CDC, if it satisfies \( |\mathcal{C}| = M \) and \( d_S(\mathcal{U}, \mathcal{V}) \geq d \) for any distinct subspaces \( \mathcal{U}, \mathcal{V} \in \mathcal{C} \), where \( |\mathcal{C}| \) represents the cardinality of \( \mathcal{C} \). Given the parameters \( n \), \( d \), \( k \) and \( q \), let \( A_q(n, d, k) \) denote the maximum size among all \((n, d, k)_q\)-CDCs. A code achieving \( A_q(n, d, k) \) is termed optimal. For CDCs, determining the exact value of \( A_q(n, d, k) \) remains a fundamental open problem.

In recent years, the significant progress has been made in constructing CDCs. Kötter and Kschischang \cite{1} pioneered the study of subspace codes in the context of error correction for random network coding. In \cite{15}, Silva et al. subsequently introduced the concept of lifting maximum rank distance (MRD) codes to construct CDCs. Nevertheless, this method often fails to achieve theoretical optimality under multiple perturbations. Given the persistent challenge of attaining optimality, more constructions for CDCs were provided subsequently to improve the lower bounds of \( A_q(n, d, k) \). 

Etzion and Silberstein \cite{2} introduced the multilevel construction, an effective method that improved the lower bound on the cardinality of CDCs by introducing identifying vectors and Ferrers diagram rank-metric codes (FDRMCs), which generalized the lifting MRD codes \cite{15}. Furthermore, the multilevel construction was extended by pending dots \cite{3,4} and pending blocks \cite{5}. Liu et al. \cite{6} then derived the parallel multilevel construction, which increased the cardinality of CDCs. Later, Liu and Ji \cite{7} presented the inverse multilevel construction and combined it with the multilevel construction to define the double multilevel construction, generalizing the parallel multilevel construction. Additionally, Yu et al. \cite{8} extended the double multilevel construction to the bilateral multilevel construction and Hong et al. \cite{26} further provided the generalized bilateral multilevel construction and presented an effective method to construct CDCs. Recently, Wang et al. \cite{wang} provided an effective construction by inserting the inverse bilateral multilevel construction into the double multilevel construction and bilateral multilevel construction.

The linkage construction is another method for constructing CDCs. Gluesing-Luerrssen and Trohn \cite{G-L} linked two shorter CDCs without compromising the subspace distance to construct longer CDCs. In \cite{9}, Cassidante et al. promoted the linkage construction based on rank-metric codes (RMCs), which improved several lower bounds on the cardinality of CDCs. The linkage construction was innovatively modified by using the coset construction by Heinlein and Kurz in \cite{co1}. Chen et al. \cite{11} further proposed parallel constructions of CDCs by using lifted MRD codes. Subsequent references \cite{ 17,25,insert1,insert2,insert3,insert4,insert5,insert6} have generalized above results, yielding enhanced lower bounds on the size of codewords of CDCs. Moreover, the connections between CDCs and mixed dimension codes (MDCs) were established in \cite{mix1}. Lao et al. utilized MDCs to construct large CDCs and He et al. \cite{mix2} further generalized this work to achieve even larger cardinality of CDCs. In addition, Li and Fu \cite{mix3} inserted more CDCs into MDCs and obtained some new lower bounds on CDCs. Recently, Xu and Song \cite{co2} proposed the list of CDCs constructed via the cosets of optimal FDRMCs, which further generalized the coset construction and increased the number of codewords for CDCs. Over the past decade, the significant research efforts have been devoted to establishing the lower bounds for cardinalities of CDCs, which are summarized in \cite{table}.

However, several current challenges persist in the study of lower bounds for CDCs: (1) insertion of additional codewords without compromising the subspace distance, (2) optimization of selection for identifying vectors set and (3) verification of optimality for FDRMC in the multilevel construction. Based on the aforementioned methodologies, the paper focuses on constructing large CDCs to address the above challenges. Firstly, we propose the construction of parallel cosets of optimal FDRMCs by employing the lists of CDCs and inverse lists of CDCs. By combining parallel cosets of optimal FDRMCs with parallel linkage construction, a new class of CDCs with large cardinality is obtained (see Theorem~\ref{thm:concatenated_cdc2}). Next, a novel set of identifying vectors is introduced, leading to a new construction of CDCs by the multilevel construction (see Theorem~\ref{th:1}). Then, under the specific conditions, we insert the coset construction into the constructions of Theorem~\ref{th:1} and propose the multilevel inserting constructions (see Theorem~\ref{th:2}). Additionally, a new class of optimal FDRMCs is proposed in Theorem~\ref{th:3} and the new optimal FDRMCs are applied in Theorem~\ref{th:1}, resulting in an alternative construction of CDCs (see Theorem~\ref{th:4}). Finally, when the parameters satisfy additional constraints, the coset construction can be inserted into the construction of Theorem~\ref{th:4}, yielding a new class of CDCs (see Theorem~\ref{th:5}). Through the above constructions, lots of lower bounds for specific CDCs can be improved.

The remainder of this paper is organized as follows. Section 2 reviews the fundamental definitions and previous results of CDCs. In Section 3, a new construction of CDCs is proposed by combining the parallel cosets of optimal FDRMCs with parallel linkage construction. Section 4 presents several new classes of CDCs based on multilevel inserting constructions. Section 5 concludes this paper.

\section{Preliminaries}
In this section, the basic definitions and known results required for our constructions will be reviewed.

\subsection{Rank Metric Codes}\label{subsec:rank_metric}
Let \( \mathbb{F}_q^{m \times n} \) denote the set of all \( m \times n \) matrices over the finite field \( \mathbb{F}_q \). For any two matrices \( A, B \in \mathbb{F}_q^{m \times n} \), the rank distance between them is defined as
\[
d_R(A, B) \triangleq \mathrm{rank}(A - B).
\]
An \emph{\((m \times n, M, \delta)_q\)} rank-metric code (RMC) is a subset \( \mathcal{R} \subseteq \mathbb{F}_q^{m \times n} \), where the rank distance between any two distinct codewords \( A, B \in \mathcal{R} \) satisfies \( d_R(A, B) \geq \delta \) and \( M = |\mathcal{R}| \). Furthermore, if \( \mathcal{R} \) forms a \( k \)-dimensional \( \mathbb{F}_q \)-linear subspace of \( \mathbb{F}_q^{m \times n} \) with minimum rank distance \( \delta \), it is denoted as a linear \([m \times n, k, \delta]_q\) RMC. 
For an \((m \times n, M, \delta)_q\) RMC, its Singleton-like upper bound is $ M \leq q^{\max\{m,n\}(\min\{m,n\} - \delta + 1)}$ \cite{13,14}. RMCs achieving this bound are called maximum rank distance (MRD) codes. The existence of linear MRD codes has been proven for all feasible parameters \cite{13,14}. 

The rank distribution of an MRD code is uniquely determined by its parameters, as stated below.

\begin{lem}[{\cite{14}}]\label{lemma:rank_dist}
Let \( m \), \( n \), \( \delta \), \( r \) be positive integers with \( \delta \leq r \leq \min\{m, n\} \). Let \( \mathcal{M} \) be a linear \([m \times n, \delta]_q\) MRD code. Define \( a(q, m, n, \delta, r) = | \{ M \in \mathcal{M} \mid \operatorname{rank}(M) = r \} | \). Then the rank distribution of \( \mathcal{M} \) is given by
\[
a(q, m, n, \delta, r) = 
\begin{bmatrix}
\min\{m, n\} \\
r
\end{bmatrix}_q 
\sum_{i=0}^{r-\delta} (-1)^i q^{\binom{i}{2}} 
\begin{bmatrix}
r \\
i
\end{bmatrix}_q 
\left( q^{\max\{m,n\}(r - i - \delta + 1)} - 1 \right).
\]
\end{lem}

The lifted MRD codes were proposed in \cite{15}. Given a linear \([k \times n, \delta]_q\) MRD code \(\mathcal{M}\) with \(n \geq k\), define
\[
\mathcal{C} = \left\{ \mathrm{rs}(I_k \mid M) \mid M \in \mathcal{M} \right\},
\]
where \(I_k\) denotes the \(k \times k\) identity matrix and \(\mathrm{rs}(I_k \mid M)\) represents a subspace of \( \mathbb{F}_q^{k \times n} \) generated by all rows of the matrix \((I_k \mid M)\). Then, \(\mathcal{C}\) is an \((n, q^{k(n - \delta + 1)}, 2\delta, k)_q\)-CDC, referred to as a lifted MRD code.

Building on the above, the classical RMC is generalized through the concept of rank-metric code with given ranks (GRMC), which is revisited as follows.

Let \( m \), \( n \), \( \delta \) be positive integers and \( K \subseteq \{0, 1, 2, \ldots, \min\{m, n\}\} \). An \((m \times n, \delta)_q\) RMC is called an \((m \times n, \delta, K)_q\)-GRMC if \( K \) contains the ranks of all codewords. Given \( m \), \( n \), \( \delta \) and \( K \), let \( A_q^G (m \times n, \delta, K) \) denote the maximum cardinality among all the \((m \times n, \delta, K)_q\)-GRMCs. For integers \( t_1 \leq t_2 \), let \([t_1, t_2]\) denote the set of integers \( x \) satisfying \( t_1 \leq x \leq t_2 \). Then, a lower bound for \( A_q^G (m \times n, \delta, [t_1, t_2]) \) is subsequently derived through Lemma 2.

\begin{lem}[{\cite{6}}]\label{lemma:GRMC_bound}
Let \( m \), \( n \), \( \delta \) be positive integers with \( \delta \leq n \leq m\). Let \(t_1\) be a nonnegative integer and \(t_2\) be a positive integer such that \(t_1 \leq t_2 \leq n\). Then
\[
A_q^G (m \times n, \delta, [t_1, t_2]) \geq 
\begin{cases} 
\displaystyle \sum_{i=t_1}^{t_2} a(q, m, n, \delta, i), & \text{if } t_2 \geq \delta; \\ 
\displaystyle \max_{\substack{\max\{1, t_1\} \leq a < \delta}} \left\lceil \frac{\sum_{i=\max\{1, t_1\}}^{t_2} a(q, m, n, a, i)}{q^{m(\delta-a)} - 1} \right\rceil, & \text{otherwise.}
\end{cases}
\]
\end{lem}

Note that when \( t_2 \geq \delta \), the set of matrices with ranks at most \( t_2 \) in an \([m \times n, \delta]_q\) MRD code forms an \((m \times n, \delta, [0, t_2])_q\)-GRMC. Therefore, there exists a lower bound such that \( A_q^G (m \times n, \delta, [0, t_2]) \geq 1 + \sum_{i=\delta}^{t_2} a(q, m, n, \delta, i) \).

\subsection{Ferrers diagrams rank-metric codes and multilevel construction}

Let \( m \) and \( n \) be positive integers. An \( m \times n \) Ferrers diagram is an \( m \times n \) array of dots and empty cells, where (1) the first row contains \( n \) dots; (2) the last column contains \( m \) dots; (3) the number of dots in each row does not exceed the number of dots in the previous row; (4) all dots are shifted to the right of the diagram.

Let \( \mathcal{F} = [\gamma_1, \gamma_2, \ldots, \gamma_n] \) denote the Ferrers diagram, where \( \gamma_i \) (\( 1 \leq i \leq n \)) represents the number of dots in the \( i \)-th column. The inverse Ferrers diagram is defined as \( \mathcal{\hat{F}} = [\gamma_n, \gamma_{n-1}, \ldots, \gamma_1] \), and the transposed Ferrers diagram is \( \mathcal{F}^t = [\rho_m, \rho_{m-1}, \ldots, \rho_1] \), where \( \rho_i \) (\( 1 \leq i \leq m \)) denotes the number of dots in the \( i \)-th row of \( \mathcal{F} \). \( \mathcal{F} \) is called empty if it contains no dots. Conversely, if \( \gamma_1 = \gamma_2 = \cdots = \gamma_n = m \), \( \mathcal{F} \) is called full.

\begin{example} Let \( \mathcal{F} = [1, 2, 4] \). Then

\[
\mathcal{F} = 
\begin{matrix}
\bullet & \bullet & \bullet \\
& \bullet & \bullet \\
& & \bullet\\
& & \bullet
\end{matrix} ,\quad
\mathcal{\hat{F}} = 
\begin{matrix}
\bullet & \bullet & \bullet\\
\bullet & \bullet & \\
\bullet & & \\
\bullet 
\end{matrix} ,\quad
\mathcal{F}^t = 
\begin{matrix}
\bullet & \bullet & \bullet & \bullet \\
& & \bullet & \bullet \\
& & & \bullet
\end{matrix} .
\]
\end{example}

Given an \( m \times n \) Ferrers diagram \( \mathcal{F} \), an \([m \times n, \rho, \delta]_q\) RMC \( \mathcal{C} \) is called an \([{\mathcal{F}}, \rho, \delta]_q\) Ferrers diagram rank-metric code (FDRMC), denoted as an \([{\mathcal{F}}, \rho, \delta]_q\) code, if, for every codeword of \( \mathcal{C} \), all entries that do not correspond to dots in the Ferrers diagram \( \mathcal{F} \) are zeros. It is noteworthy that if there exists an \([{\mathcal{F}}, \rho, \delta]_q\) code, there also exists an \([\mathcal{\hat{F}}, \rho, \delta]_q\) code and an \([{\mathcal{F}}^t, \rho, \delta]_q\) code \cite{7}. The following lemma provides a Singleton-like bound for FDRMCs.

\begin{lem}[{\cite{2}}]{\label{lem:op}}
Let \( \mathcal{F} \) be an \( m \times n \) Ferrers diagram. For \( 0 \leq i \leq \delta - 1 \), let \( v_i \) denote the number of dots in the submatrix obtained by removing the first \( i \) rows and the rightmost \( \delta - 1 - i \) columns. Then, for any \([{\mathcal{F}}, \rho, \delta]_q\) code, the dimension \( \rho \) satisfies
\(
\rho \leq \min \big\{ v_i \big\}_{0 \leq i \leq \delta - 1}.
\)
\end{lem}

An FDRMC is called an optimal FDRMC if it attains the bound of Lemma~\ref{lem:op}. The construction of optimal FDRMCs has been extensively studied (\cite{18,19,20,21,22,23}). The following lemma presents a method for verifying the existence of optimal FDRMCs.

\begin{lem}[{\cite{19}}]{\label{lem:op_FDRMC}}
Given an \( m \times n \) Ferrers diagram \( \mathcal{F} \) with \( m \geq n \), if the last \( \delta - 1 \) columns of \( \mathcal{F} \) contain at least \( m \) dots, then there exists an optimal \([{\mathcal{F}}, \sum_{i=1}^{n - \delta + 1} \gamma_i, \delta]_q\) code,
where \( \gamma_i \) denotes the number of dots in the \( i \)-th column of \( \mathcal{F} \).
\end{lem}

Note that when \( \delta \leq 2 \), there always exists an optimal \([{\mathcal{F}}, \delta]_q\) code for any Ferrers diagram \( \mathcal{F} \) and prime power \( q \).

\begin{lem}[{\cite{19}}]
Let \( \mathcal{F} \) be an \( m \times m \) Ferrers diagram. For any prime power \( q \), there exists an optimal \([{\mathcal{F}}, \delta]_q\) code.
\end{lem}

In \cite{19}, two methods for constructing new FDRMCs based on RMC are proposed. The following lemma describes one of the two methods, which combines two FDRMCs of the same dimension. This method is crucial for our subsequent construction.

\begin{lem}[{\cite{19}}]\label{lem:composite_FDRMC}
Let \(\mathcal{F}_1\) be an \( m_1 \times n_1 \) Ferrers diagram and assume \(\mathcal{C}_1\) is an \([{\mathcal{F}_1}, \rho, \delta_1]_q\) code; let \(\mathcal{F}_2\) be an \( m_2 \times n_2 \) Ferrers diagram and assume \(\mathcal{C}_2\) is an \([{\mathcal{F}_2}, \rho, \delta_2]_q\) code; let \(\mathcal{D}\) be an \( m_3 \times n_3 \) full Ferrers diagram, where \( m_3 \geq m_1 \) and \( n_3 \geq n_2 \).
Let
\[
\mathcal{F} = 
\begin{pmatrix}
\mathcal{F}_1 & \mathcal{D} \\
& \mathcal{F}_2 
\end{pmatrix}
\]
is an \( m \times n \) Ferrers diagram with \( m = m_2 + m_3 \) and \( n = n_1 + n_3 \). Then there exists an \([{\mathcal{F}}, \rho, \delta_1+\delta_2]_q\) code.
\end{lem}

Given an \( m \times n \) Ferrers diagram \( \mathcal{F} \), an \([{\mathcal{F}}, \delta]_q\) code \( \mathcal{C} \) is called an \((\mathcal{F}, \delta, [0, r])_q\)-Ferrers diagram rank-metric code with given ranks (GFRMC), if every codeword of \(\mathcal{C} \) has rank at most \( r \). For fixed parameters \( m, n, \delta \) and \( r \), let \( A_q^F( {\mathcal{F}}, \delta, r) \) denote the maximum cardinality among all \(({\mathcal{F}}, \delta, [0, r])_q\)-GFRMCs. The following lemma establishes a bound for \( A_q^F( \mathcal{F}, \delta, r) \).

\begin{lem}[{\cite{7}}]\label{lem:FDRMC_lower_bound}
Let \( n \geq r \) and \( \mathcal{F} = [\gamma_1, \gamma_2, \ldots, \gamma_n] \) be an Ferrers diagram. Then
\[
A_q^F( \mathcal{F}, \delta, r) \geq \max_{1 \leq i \leq n} A_q^G(\gamma_i \times (n-i), \delta, [0, r]).
\]
\end{lem}

A matrix is said to be in reduced row (row inverse) echelon form, abbreviated as RR(I)EF, if it satisfies the following conditions: (1) the leading coefficient of each nonzero row is strictly to the right (left) of the leading coefficient in the row above it; (2) all leading coefficients are equal to 1; (3) each leading coefficient is the only nonzero entry in its column.

A \( k \)-dimensional subspace \( \mathcal{U} \) of \( \mathbb{F}_q^n \) can be represented by a generator matrix \( U \in \mathbb{F}_q^{k \times n} \), where the rows of \( U \) form a basis for the subspace.  Note that the generator matrice \( U \) is not unique for \( \mathcal{U} \). For any \( \mathcal{U} \in \mathcal{G}_q(n, k) \), there exists a unique matrix in RREF and a unique matrix in RRIEF, denoted by \( E(\mathcal{U}) \) and \( \hat{E}(\mathcal{U}) \), respectively. The identifying vector \( v(\mathcal{U}) \) is a binary row vector of length \( n \) and weight \( k \), where the positions of the ones correspond to the pivot columns of \( E(\mathcal{U}) \). The echelon Ferrers form EF(\( v(\mathcal{U}) \)), abbreviated as EF(\( v \)), is defined as a matrix in RREF with leading entries (of rows) in the columns indexed by the nonzero entries of \( v(\mathcal{U}) \) and \(\bullet\) in all entries which do not have terminals zeros or ones. The Ferrers diagram \( \mathcal{F}_v \) is derived through the following operations applied to  EF($v$): (1) remove all zeros to the left of the pivot in each row; (2) remove all columns containing pivots. (3) shift the remaining entries to the right. Similarly, given a generator matrix in RRIEF of a subspace, we can obtain the inverse identifying vector \( \hat{v} \), the inverse echelon Ferrers form \(\widehat{\text{EF}}( \hat{v} )\) and inverse Ferrers diagram \( \hat{\mathcal{F}}_{\hat{v}} \).

\begin{example}
Consider a 3-dimensional subspace \( \mathcal{U} \in \mathcal{G}_2(5, 3) \) with generator matrices in RREF and RRIEF being
\[
E(\mathcal{U}) = 
\begin{bmatrix}
\mathbf{1} & 0 & 0 & 0 & 1 \\
0 & 0 & \mathbf{1} & 0 & 1 \\
0 & 0 & 0 & \mathbf{1} & 0
\end{bmatrix}, \quad 
\hat{E}(\mathcal{U}) = 
\begin{bmatrix}
1 & 0 & 0 & 0 & \mathbf{1} \\
0 & 0 & 0 & \mathbf{1} & 0 \\
1 & 0 & \mathbf{1} & 0 & 0
\end{bmatrix}.
\]
The identifying vector \( v \) of \( \mathcal{U} \) is \( (10110) \), and the inverse identifying vector \( \hat{v} \) of \( \mathcal{U} \) is \( (00111) \). The corresponding echelon Ferrers forms are
\[
EF( v ) =
\begin{bmatrix}
\mathbf{1} & \bullet & 0 & 0 & \bullet \\
0 & 0 & \mathbf{1} & 0 & \bullet \\
0 & 0 & 0 & \mathbf{1} & \bullet
\end{bmatrix}, \quad 
EF(\hat{v}) = 
\begin{bmatrix}
\bullet & \bullet & 0  & 0 & \mathbf{1} \\
\bullet & \bullet & 0  & \mathbf{1} & 0 \\
\bullet & \bullet & \mathbf{1}  & 0 & 0
\end{bmatrix}.
\]
The resulting Ferrers diagram \( \mathcal{F}_{v} \) and its inverse Ferrers diagram \( \hat{\mathcal{F}}_{\hat{v}} \) are
\[
\mathcal{F}_{v} = 
\begin{matrix}
\bullet & \bullet  \\
 & \bullet  \\
 & \bullet 
\end{matrix}, \quad 
\hat{\mathcal{F}}_{\hat{v}}  = 
\begin{matrix}
\bullet & \bullet   \\
\bullet & \bullet   \\
\bullet & \bullet  
\end{matrix}.
\]
\end{example}

The following lemmas present lower bounds on subspace distance by Hamming distance of vectors.

\begin{lem}[{\cite{2}}]\label{lemma:subspace_hamming}
Let \( \mathcal{U}, \mathcal{V} \in \mathcal{G}_q(n, k) \), where \( \mathcal{U} = rs(U) \) and \( \mathcal{V} = rs(V) \), with \( U, V \in \mathbb{F}_n^{k \times n} \) in RREFs. Let \( C_U \) and \( C_V \) denote the submatrices of \( U \) and \( V \), respectively, formed by removing the pivot columns. Then
\[
d_S(\mathcal{U}, \mathcal{V}) \geq d_H((v(\mathcal{U}), v(\mathcal{V})).
\]
If \( v(\mathcal{U}) = v(\mathcal{V}) \), it holds that
\[
d_S(\mathcal{U}, \mathcal{V}) = 2 \cdot d_R(C_U, C_V).
\]
\end{lem}

\begin{lem}[{\cite{7}}]\label{lemma:subspace_hamming2}
Let \( \mathcal{U}, \mathcal{V} \in \mathcal{G}_q(n, k) \), where \( \mathcal{U} = rs(U) \) and \( \mathcal{V} = rs(V) \) with \( U, V \in \mathbb{F}_n^{k \times n} \) in RRIEFs. Let \( C_U \) and \( C_V \) denote the submatrices of \( U \) and \( V \), respectively, formed by removing the pivot columns. Then
\[
d_S(\mathcal{U}, \mathcal{V}) \geq d_H((\hat{v}(\mathcal{U}), \hat{v}(\mathcal{V})).
\]
If \( \hat{v}(\mathcal{U}) = \hat{v}(\mathcal{V}) \), it holds that
\[
d_S(\mathcal{U}, \mathcal{V}) = 2 \cdot d_R(C_U, C_V).
\]
\end{lem}

A set \( \mathcal{S} \) of binary vectors is called an \((n, 2\delta, k)_2\) constant weight code (CWC) if it satisfies: (1) every vector in \( \mathcal{S} \) is a binary vector of length \( n \) with weight \( k \); (2) the Hamming distance between any two distinct vectors in \( \mathcal{S} \) is at least \( 2\delta \). Leveraging aforementioned results, the multilevel construction is detailed below.

\begin{lem}[Multilevel construction \cite{2}]\label{lemma:multilevel}  
    Let \( \mathcal{S} \) be an \((n, 2\delta, k)_2\) CWC. Suppose that for every vector \( v \in \mathcal{S} \), its echelon Ferrers form is $EF(v)$. All dots in $EF(v)$ form an Ferrers diagram \( \mathcal{F}_v \), and there exists an \([ \mathcal{F}_v, \delta ]_q\) code \( \mathcal{C}_v \), then \( \cup_{v \in \mathcal{S}} \mathcal{C}_v \) constitutes an \((n, 2\delta, k)_q\)-CDC.  
\end{lem}  

\subsection{Lists of CDCs with fixed distance and coset constuction}

Given two CDCs \( \mathscr{C}_1, \mathscr{C}_2 \in \mathcal{G}_q(n, k) \), the subspace distance between them is defined as
\[
d_S(\mathscr{C}_1, \mathscr{C}_2) = \text{min}\{{d_S(c_1, c_2)|c_1 \in \mathscr{C}_1, c_2 \in \mathscr{C}_2}\}.
\]  
The definition of list of CDCs with fixed distance is as follows.

\begin{definition}[\cite{co2}] 
Let \( n, \delta_1, \delta_2, k, s \) be integers with \( \delta_1 > \delta_2 > 0\). We say that \( (\mathscr{C}_1, \mathscr{C}_2, \cdots, \mathscr{C}_s) \) is a list of \((n, 2\delta_1, k)_q\)-CDCs with fixed distance \( 2\delta_2 \), if each \( \mathscr{C}_i \) is an \((n, 2\delta_1, k)_q\)-CDC for \( 1\leq i \leq s\) and \( d_S(\mathscr{C}_i, \mathscr{C}_j) \geq 2\delta_2\) for \( 1\leq i < j \leq s\).
\end{definition}

In \cite{co2}, Xu and Song promoted a construction for list of CDCs via cosets of optimal FDRMCs. 

\begin{lem}[\cite{co2}]\label{lemma:optimal_subspace}  
Let \( m, n, \delta_1, \delta_2 \) be integers with \( m \geq n \geq \delta_1 > \delta_2 > 0\). Given an \( m \times n \) Ferrers diagram \( \mathcal{F} = [\gamma_1, \gamma_2, \cdots, \gamma_n] \), suppose that each of the rightmost \( \delta_1 - 1 \) columns of \( \mathcal{F} \) contain at least \( n \) dots. Then there exists an optimal \( [\mathcal{F}, \delta_1]_q \) code \( \mathcal{C}^1 \) and an optimal \( [\mathcal{F}, \delta_2]_q \) code \( \mathcal{C}^2 \), such that \( \mathcal{C}^1 \) is an \( \mathbb{F}_q \)-subspace of \( \mathcal{C}^2 \). When \( \mathcal{C}^2 \) is an \( [\mathcal{F}, 1]_q \) code, \( \mathcal{C}^2 \) is optimal. Therefore, any optimal \( [\mathcal{F}, \delta_1]_q \) code \( \mathcal{C}^1 \) is an \( \mathbb{F}_q \)-subspace of \( \mathcal{C}^2 \). Furthermore, if \( v_{\text{min}}(\mathcal{F}, \delta_1) = 0 \), define \( \mathcal{C}^1 = \{\mathbf{0}_{m \times n}\} \) as the optimal \( [\mathcal{F}, \delta_1]_q \) code. Therefore, \( \mathcal{C}^1 \) is an \( \mathbb{F}_q \)-subspace of any optimal \( [\mathcal{F}, \delta_2]_q \) code \( \mathcal{C}^2 \).
\end{lem}

By Lemma~\ref{lemma:optimal_subspace}, a construction of list of CDCs with fixed distance by cosets of an optimal
FDRMC is given below.

\begin{lem}[\cite{co2}]\label{lemma:cpc_sequence}  
Let \( \mathcal{C}^1 \) and \( \mathcal{C}^2 \) denote two optimal FDRMCs defined in Lemma~\ref{lemma:optimal_subspace}, where \( \mathcal{C}^1 \) is an \( \mathbb{F}_q \)-subspace of \( \mathcal{C}^2 \). Then there exists a list of \((n, D_v, 2\delta_1, k)_q\)-CDCs \( (\mathscr{C}_1, \mathscr{C}_2, \cdots, \mathscr{C}_{s_v}) \) with fixed distance \( 2\delta_2 \), where \(D_v = q^{v_{\text{min}}(\mathcal{F}_v, \delta_1)}, s_v = q^{v_{\text{min}}(\mathcal{F}_v, \delta_2) - v_{\text{min}}(\mathcal{F}_v, \delta_1)}\). 
\end{lem}

Based on Lemma~\ref{lemma:optimal_subspace} and Lemma~\ref{lemma:cpc_sequence}, the construction of inverse list of CDCs with fixed distance by cosets of an optimal FDRMC is given as follows.

\begin{coro}[]\label{cor:codes_sequence}  
Let \( n, k, \delta_1, \delta_2 \) be integers with \( \delta_1 > \delta_2 > 0 \). Given an inverse identifying vectors \( \hat{v} \in \mathcal{F}^n_2 \) of weight \( k \), let \( \hat{\mathcal{F}}_{\hat{v}} \) denote its associated inverse Ferrers diagram. Suppose that each of the leftmost \( \delta_1 - 1 \) columns of \( \hat{\mathcal{F}}_{\hat{v}} \) contain at least \( n \) dots. Then 
\begin{itemize}  
    \item there exists an optimal \( [\hat{\mathcal{F}}_{\hat{v}}, \delta_1] \) code \( \mathcal{C}^1 \) and an optimal \( [\hat{\mathcal{F}}_{\hat{v}}, \delta_2] \) code \( \mathcal{C}^2 \), such that \( \mathcal{C}^1 \) is an \( \mathbb{F}_q \)-subspace of \( \mathcal{C}^2 \);  
    \item there exists an inverse list of \((n, D_{\hat{v}}, 2\delta_1, k)_q\)-CDCs \((\mathscr{C}_1, \mathscr{C}_2, \cdots, \mathscr{C}_{s_{\hat{v}}})\) with fixed distance \(2\delta_2\), where \( D_{\hat{v}} = q^{v_{\text{min}}(\hat{\mathcal{F}}_{\hat{v}}, \delta_1)}, s_{\hat{v}} = q^{v_{\text{min}}(\hat{\mathcal{F}}_{\hat{v}}, \delta_2) - v_{\text{min}}(\hat{\mathcal{F}}_{\hat{v}}, \delta_1)} \).
\end{itemize}   
\end{coro}

Keep the notations as in Corollary~\ref{cor:codes_sequence} and let \( \mathcal{C}^2 / \mathcal{C}^1 = \{ \mathcal{D}_1, \mathcal{D}_2, \cdots, \mathcal{D}_{s_{\hat{v}}} \} \) be the complete set of cosets of \(\mathcal{C}^1\) in \(\mathcal{C}^2\). For each \( \mathcal{D}_i \), \( 1\leq i \leq s_{\hat{v}}\), if there exists an \((\hat{\mathcal{F}}_{\hat{v}}, \delta_1, [0, r])_q\)-GFRMC $\mathcal{C}_i$ and let $\mathscr{C}_i$ be the corresponding lifted code, then \((\mathscr{C}_1, \mathscr{C}_2, \cdots, \mathscr{C}_{s_{\hat{v}}})\) is defined as a $r$-restricted inverse list of \((n, 2\delta_1, k)_q\)-CDCs with fixed distance \(2\delta_2\).

The following lemma can be used to construct larger CDCs through lists of CDCs.

\begin{lem}[\cite{co2}]\label{lemma:cdc_fdrmc}  
Let \(n, k, t, \delta_1, \delta_2\) be integers with \(\delta_1 > \delta_2 > 0\). Let \(\mathcal{V}\) be an \((n, 2\delta_1, k)_2\) CWC such that for each \(v \in \mathcal{V}\), there exists an optimal \([\mathcal{F}_v, \delta_1]_{q}\) code \(\mathcal{C}^1\) which is an \(\mathbb{F}_{q}\)-subspace of an optimal \([\mathcal{F}_v, \delta_2]_{q}\) code \(\mathcal{C}^2\). Suppose that \(\mathcal{V} = \{v_1, v_2, \cdots, v_t\}\) with \(s_{v_1} \geq s_{v_2} \geq \cdots \geq s_{v_t} > s_{v_{t+1}} = 0 \). Then, there exists \(s_{v_i} - s_{v_{i+1}}\) distinct \((n, \sum_{k=1}^{i} D_{v_k}, 2\delta_1, k)_{q}\)-CDCs for \(1 \leq i \leq t\). Moreover, all these CDCs collectively form a list of \((n, 2\delta_1, k)_{q}\)-CDCs with fixed distance \(2\delta_2\).  
\end{lem}

Based on Lemma~\ref{lemma:cdc_fdrmc}, the following lemma is presented, which is essential to our construction.

\begin{lem}[\cite{co2}]\label{lemma:cascaded_codes}  
Let \(n, k, t, \delta_1, \delta_2\) be integers with \(\delta_1 > \delta_2 > 0\). Let \( \mathcal{U} = \bigcup_{i=1}^t \mathcal{V}_i \) be an \((n, 2\delta_2, k)_2\) CWC, where each \( \mathcal{V}_i \) is an \((n, 2\delta_1, k)_2\) CWC satisfying the conditions of Lemma~\ref{lemma:cdc_fdrmc}. Suppose \( \mathscr{C}_{v_i} \) is a list of CDCs constructed via Lemma~\ref{lemma:cdc_fdrmc}. Then, \( (\mathscr{C}_{v_1}, \mathscr{C}_{v_2}, \cdots, \mathscr{C}_{v_t}) \) forms a list of \((n, 2\delta_1, k)_q\)-CDCs with fixed distance \(2\delta_2\).  
\end{lem}

\begin{rem}\label{remark}

Note that for a (restricted) inverse list of CDCs, Lemma~\ref{lemma:cdc_fdrmc} and Lemma~\ref{lemma:cascaded_codes} remain valid.  

\end{rem}

The coset construction proposed in \cite{co1} is described below. Given a \( k \)-dimensional subspace \( \mathcal{U} \) of \( \mathbb{F}_q^n \) and a generator matrix \( U \in \mathbb{F}_q^{k \times n} \) in RREF. Then there exists a mapping 
\[
\tau: \mathcal{G}_q(n, k) \rightarrow \left\{ U \in \mathbb{F}_q^{k \times n} \,\big|\, \mathrm{rank}(U) = k, \, U \text{ in RREF} \right\}.
\]  
Let \( B \in \mathbb{F}_q^{k \times n} \) be a full rank matrix in RREF and \( F \in \mathbb{F}_q^{y \times (n-k)} \) with \( n \geq k \). Let \( b_i \) denote the pivot column positions of \( B \) for \( 1 \leq i \leq k \). Suppose \( \varphi_B(F) \in \mathbb{F}_q^{y \times n} \) satisfies:  
\begin{itemize}  
    \item if \( i \in \{b_1, \ldots, b_k\} \), then \( g_i = \mathbf{0} \in \mathbb{F}_q^y \);  
    \item the remaining submatrix after removing columns \( g_{b_1}, \cdots, g_{b_k} \) from \( \varphi_B(F) \) is \( F \),
\end{itemize}   
where \( g_i \) is the \( i \)-th column vector of \( \varphi_B(F) \) for \( 1 \leq i \leq n \).  

\begin{example}
Let
\[
 B = \begin{bmatrix} 
1 & 1 & 0 & 0 & 1 & 0 \\ 
0 & 1 & 0 & 1 & 1 & 1 \\ 
0 & 0 & 0 & 0 & 1 & 1 
\end{bmatrix}, 
F = \begin{bmatrix} 
1 & 1 & 0 \\ 
0 & 0 & 0 \\ 
1 & 0 & 1 \\ 
0 & 0 & 1 
\end{bmatrix}.
\]  
Then
\[
\varphi_B(F) = \begin{bmatrix} 
0 & 0 & 1 & 1 & 0 & 0 \\ 
0 & 0 & 0 & 0 & 0 & 0 \\ 
0 & 0 & 1 & 0 & 0 & 1 \\ 
0 & 0 & 0 & 0 & 0 & 1 
\end{bmatrix}.
\]  
\end{example}

\begin{lem}[Coset construction \cite{co1}]\label{lemma:coset_construction}  
    Let \( n, d, k, n_1, n_2, d_1, d_2, k_1, k_2, s \) be positive integers with \( n = n_1 + n_2 \), \( d = d_1 + d_2 \) and \( k = k_1 + k_2 \), where \( n \geq 2k \) and \( n_i \geq k_i \) for \( i = 1, 2\). Suppose \( (\mathcal{A}_1, \mathcal{A}_2, \cdots, \mathcal{A}_s) \) is a list of \((n_1, d, k_1)_q\)-CDCs with fixed distance \( d_1 \), \( (\mathcal{B}_1, \mathcal{B}_2, \cdots, \mathcal{B}_s) \) is a list of \((n_2, d, k_2)_q\)-CDCs with fixed distance \( d_2 \) and \( \mathcal{H} \) is a \( (k_1 \times (n_2-k_2), \frac{d}{2})_q\) MRD code. Define
    \[
    \mathcal{C}^i = \left\{ \mathrm{rs}\left(
    \begin{array}{@{}cc@{}}
    A_i & \varphi_{B_i}(H) \\
    0_{k_2 \times n_1} & B_i \\
    \end{array}
    \right)\,\bigg|\,  
    \tau^{-1}(A_i) \in \mathcal{A}_i, \, \tau^{-1}(B_i) \in \mathcal{B}_i, \, H \in \mathcal{H}  
    \right\}.
    \]  
Then \( \bigcup_{1 \leq i \leq s} \mathcal{C}^i \) is an \((n, d, k)_q\)-CDC.  
\end{lem}

The following lemma is crucial for constructions of larger CDCs by using the multilevel construction and coset construction.

\begin{lem}[\cite{co1}]\label{lemma:multi_coset}  
Keep the notations as in Lemma~\ref{lemma:coset_construction}. For \( \mathcal{U} \in \bigcup_{i \leq i \leq s} \mathcal{C}^i \) and \( \mathcal{V} \in \mathcal{G}_{q}(n,k) \), let \( x \) be the partial sum of the first \( n_1 \) entries of the identifying vector \( v(\mathcal{V}) \). If \( d \leq 2|x - k_1| \), then the subspace distance satisfies \(d_S(\mathcal{U}, \mathcal{V}) \geq d \).
\end{lem}

The following lemma will be used throughout the remainder of this paper.

\begin{lem}[\cite{co1}]\label{lemma:rearrangement}  
    Let \( a_1 \geq a_2 \geq \cdots \geq a_s \) and \( b_1 \geq b_2 \geq \cdots \geq b_s \) be positive integers. Then, for any permutation \( \beta: \{1, 2, \ldots, s\} \to \{1, 2, \ldots, s\} \), the following inequality holds 
    \[
    \sum_{i=1}^s a_i b_i \geq \sum_{i=1}^s a_i b_{\beta(i)}.
    \]  
\end{lem}

He \cite{17} introduced the following lemma that generalized the parallel linkage construction and effectively improved lower bounds on the cardinality of CDCs.

\begin{lem}[\cite{17}]\label{lemma:combined_cdc}  
Let \( n, d, k, n_1, n_2 \) be positive integers, where \( n = n_1 + n_2 \) and \( n_i \geq k \geq \frac{d}{2} \) for \( i = 1, 2 \). Suppose \( \mathcal{U}_1 \) is an \((n_1, d, k)_q\)-CDC, \( \mathcal{U}_2 \) is an \((n_2, d, k)_q\)-CDC, \( \mathcal{M}_1 \) is a \([k \times n_2, \frac{d}{2}]_q\) MRD code and \( \mathcal{M}_2 \) is a \((k \times n_1, \frac{d}{2}, [0, k-\frac{d}{2}])_q\)-GRMC. Define 
    \[
    \mathcal{C}_1 = \left\{ \mathrm{rs}(U_1 \mid M_1) \,\big|\, \tau^{-1}(U_1) \in \mathcal{U}_1, M_1 \in \mathcal{M}_1 \right\}, 
\]
\[ 
\mathcal{C}_2 = \left\{\mathrm{rs}(M_2 \mid U_2) \,\big|\, \tau^{-1}(U_2) \in \mathcal{U}_2, M_2 \in \mathcal{M}_2 \right\}.
    \]  
Then \( \mathcal{C}_1 \cup \mathcal{C}_2 \) is an \((n, d, k)_q\)-CDC.  
\end{lem}

\section{Constructions of CDCs by parallel cosets of optimal FDRMCs}\label{sec:cdc_fdrmc}  
In this section, we first propose the construction of parallel cosets of optimal FDRMCs. Then, a new class of CDCs are presented by combining the construction of parallel cosets of optimal FDRMCs with parallel linkage construction. The construction improves the lower bound on $A_q(18,8,9)$. Throughout the remainder of this section, \( \mathcal{C}_1 \) and \( \mathcal{C}_2 \) denote the CDCs defined in Lemma~\ref{lemma:combined_cdc}. 

\begin{theo}[Parallel cosets of optimal FDRMCs]\label{thm:concatenated_cdc}  
Let \( n, d, k, n_1, n_2, d_1, d_2, k_1, k_2, s_1, s_2\) be positive integers with \( n = n_1 + n_2, k = k_1 + k_2 \) and \( d = d_1 + d_2\), where \( n_i \geq k_i\) and \( k_i \geq \frac{d}{2}\) for \(i=1,2 \). Suppose \( (\mathcal{A}_i)_{1 \leq i \leq s_1} \) is a list of \((n_1, d, k_1)_q\)-CDCs with fixed distance \(d_1\),                         \( (\mathcal{B}_i)_{1 \leq i \leq s_2} \) is a list of \((n_2, d, k_2)_q\)-CDCs with fixed distance \(d_2\). For each \( 1 \leq i \leq s_1 \), define   
    \[
    \mathcal{C}_3^i = \left\{ \mathrm{rs}\left( 
    \begin{array}{@{}cc@{}}
    A_i & 0_{k_1 \times n_2} \\
    0_{k_2 \times n_1} & B_i \\
    \end{array}
    \right) \,\bigg|\, 
    \tau^{-1}(A_i) \in \mathcal{A}_i, \, \tau^{-1}(B_i) \in \mathcal{B}_i
    \right\},
    \]  
    where \(A_i\) and \(B_i\) in RREFs. Then \(\mathcal{C}_3 = \bigcup_{i=1}^{s_1} \mathcal{C}_3^i \) is an \((n, d, k)_q\)-CDC.\\
 \indent  Suppose \( (\hat{\mathcal{A}_j})_{1 \leq j \leq s_2} \) is a $r_{\hat{u}_2}$-restricted inverse list of \((n_1, d, k_1)_q\)-CDCs with fixed distance \( d_1 \), and \( (\hat{\mathcal{B}_j})_{1 \leq j \leq s_2} \) is an inverse list of \((n_2, d, k_2)_q\)-CDCs with fixed distance \( d_2 \). For each \( 1 \leq j \leq s_2 \), define
    \[
    \mathcal{C}_4^j = \left\{ \mathrm{rs}\left( 
    \begin{array}{@{}cc@{}}
    0_{k_2 \times n_1} & \hat{B_j} \\
    \hat{A_j} & 0_{k_1 \times n_2} \\
    \end{array}
    \right) \,\bigg|\, 
    \tau^{-1}(\hat{A_j}) \in \hat{\mathcal{A}_j}, \, \tau^{-1}(\hat{B_j}) \in \hat{\mathcal{B}_j} \,
    \right\},
    \]  
    where \(\hat{A_j}\) and \(\hat{B_j}\) in RRIEFs. Then \(\mathcal{C}_4 = \bigcup_{j=1}^{s_2} \mathcal{C}_4^j \) is an \((n, d, k)_q\)-CDC, moreover, \( \mathcal{C}_3 \cup \mathcal{C}_4\) is an \((n, d, k)_q\)-CDC.
\end{theo}

\begin{proof}
As shown in Theorem 3.2 of \cite{co2}, $\mathcal{C}_3$ is proven to be an $(n, d, k)_q$-CDC. First, we prove that $\mathcal{C}_4$ is an $(n, d, k)_q$-CDC. It can be directly verified that all the codewords of $\mathcal{C}_4$ are $k$-dimensional subspaces of $\mathbb{F}_q^{n}$. Take any two distinct codewords $Y_1, Y_2 \in \mathcal{C}_4$ and let
\[
Y_1 = \mathrm{rs}\begin{pmatrix} 0&\hat{B}\\ \hat{A}&0 \end{pmatrix}, \quad 
Y_2 = \mathrm{rs}\begin{pmatrix} 0&\hat{B}'\\ \hat{A}'&0 \end{pmatrix},
\]
where there exist $1 \leq i,j \leq s_2$ such that $\tau^{-1}(\hat{A}) \in \hat{\mathcal{A}_i},\tau^{-1}(\hat{B}) \in \hat{\mathcal{B}_i},\tau^{-1}(\hat{A}') \in \hat{\mathcal{A}_j},\tau^{-1}(\hat{B}') \in \hat{\mathcal{B}_j}$. Then
\[
\begin{split}
	d_S(Y_1, Y_2)
	&= 2\mathrm{rank}\left( 
	\begin{array}{@{}cc@{}}
		0&\hat{B} \\
		\hat{A}&0 \\
		0&\hat{B}'\\
		\hat{A}'&0 \\
	\end{array}
	\right)-2k\\
	&= 2\mathrm{rank}\left( 
	\begin{array}{@{}cc@{}}
		\hat{A}\\
		\hat{A}'\\
	\end{array}
	\right)+ 2\mathrm{rank}\left( 
	\begin{array}{@{}cc@{}}
		\hat{B}\\
		\hat{B}'\\
	\end{array}
	\right)-2k\\
	&= d_S(\hat{A},\hat{A}')+2k_1+d_S(\hat{B},\hat{B}')+2k_2-2k\\
	&\geq d_1+d_2\\
	&=d.
\end{split}
\]
So $\mathcal{C}_4$ is an $(n, d, k)_q$-CDC.

Next, we prove that \( \mathcal{C}_3 \cup \mathcal{C}_4 \) is an \((n, d, k)_q\)-CDC. Let
\[
W = \mathrm{rs}\begin{pmatrix} A & 0 \\ 0 & B \end{pmatrix} \in \mathcal{C}_3,
Y = \mathrm{rs}\begin{pmatrix} 0&\hat{B}\\ \hat{A}&0 \end{pmatrix}\in \mathcal{C}_4.
\]
where there exist $1 \leq i \leq s_1, 1 \leq j \leq s_2$ such that $\tau^{-1}(A) \in \mathcal{A}_i,\tau^{-1}(B) \in \mathcal{B}_i,\tau^{-1}(\hat{A}) \in \hat{\mathcal{A}_j},\tau^{-1}(\hat{B}) \in \hat{\mathcal{B}_j}$. Then
\[
\begin{split}
	d_S(W, Y)
	&= 2\mathrm{rank}\left( 
	\begin{array}{@{}cc@{}}
		A & 0 \\
		0 & B \\
		0 & \hat{B}\\
		\hat{A} & 0 \\
	\end{array}\right)-2k\\
	&\geq 2\mathrm{rank}\left( 
	\begin{array}{@{}cc@{}}
		A & 0 \\
		0 & B\\
		\hat{A} & 0\\
	\end{array}
	\right)-2k\\
	&= 2\mathrm{rank}\left( 
	\begin{array}{@{}cc@{}}
		A  \\
		\hat{A}\\
	\end{array}
	\right)+2\mathrm{rank}(B)-2k\\
	&= 2\mathrm{rank}\left( 
	\begin{array}{@{}cc@{}}
		A_1 & A_2  \\
		\hat{A}_3 & I_{k_1}\\
	\end{array}
	\right)-2k_1\\
	&= 2\mathrm{rank}\left( 
	\begin{array}{@{}cc@{}}
		A_1-\hat{A}_3A_2 & 0 \\
		\hat{A}_3 & I_{k_1}\\
	\end{array}
	\right)-2k_1\\
  &\geq 2\mathrm{rank}\left( 
	\begin{array}{@{}cc@{}}
		A_1-\hat{A}_3A_2  \\
	\end{array}
	\right)\\
	&\geq 2\mathrm{rank}(A_1)-2\mathrm{rank}(\hat{A}_3)\\
	&\geq 2(r_{\hat{u}_2}+\frac{d}{2})-2r_{\hat{u}_2}\\
	&=d,
\end{split}
\]
where $\left( 
\begin{array}{@{}cc@{}}
	A_1 & A_2  \\
	\hat{A}_3 & I_{k_1}\\
\end{array}
\right)$ represents $\left( 
\begin{array}{@{}cc@{}}
	A  \\
	\hat{A}\\
\end{array}
\right)$ after elementary row and column transformations. Therefore, \( \mathcal{C}_3 \cup \mathcal{C}_4\) is an \((n, d, k)_q\)-CDC.
\end{proof}

Under the certain conditions, combining the parallel cosets of optimal FDRMCs with parallel linkage construction yields a new  construction of CDCs as follows.

\begin{theo}\label{thm:concatenated_cdc2}  
Keep the notations as in Theorem~\ref{thm:concatenated_cdc}. Let \( \mathcal{U}_1 \in \bigcup_{i=1}^{s_1} \mathcal{A}_i \), \( \mathcal{U}_2 \in \bigcup_{j=1}^{s_2} \hat{\mathcal{A}}_j \) and \(u_1, \hat{u}_2\) be identifying vectors of \( \mathcal{U}_1\) and \( \mathcal{U}_2\), respectively. Furthermore, suppose $\mathcal{U}_2=rs(U_2)$ and rank$(U_2)\leq r_{\hat{u}_2}$. If
    \[
    d_H(u_1, \hat{u}_2) \geq 2\left(r_{\hat{u}_2} + \frac{d}{2}\right),
    \]  
then \( \mathcal{C}_1 \cup \mathcal{C}_2 \cup \mathcal{C}_3 \cup \mathcal{C}_4\) is an \((n, d, k)_q\)-CDC and
\[
\begin{split}
	A_q(n, d, k) \geq &A_q(n_1, d, k) \cdot q^{n_2 (k- \frac{d}{2}+1)} + \left( \sum_{i=1}^{k- \frac{d}{2}}a(q, m, n, \frac{d}{2}, i) + 1 \right) \cdot A_q(n_2, d, k)\\ 
	&+ \sum_{i=1}^{s_1} |\mathcal{A}_i| \cdot |\mathcal{B}_i| + \sum_{j=1}^{s_2} |\hat{\mathcal{A}}_j| \cdot |\hat{\mathcal{B}}_j|.
\end{split}
\]   
\end{theo}

\begin{proof}
By Theorem 3.2 of \cite{co2}, $\mathcal{C}_1 \cup \mathcal{C}_2 \cup \mathcal{C}_3 $ is an $(n, d, k)_q$-CDC. First, it is proved that $\mathcal{C}_1 \cup \mathcal{C}_2 \cup \mathcal{C}_4$ is an $(n, d, k)_q$-CDC. Let
\[
V = \mathrm{rs}\begin{pmatrix} U_1 \mid M_1  \end{pmatrix}\in \mathcal{C}_1, \quad 
X=\mathrm{rs}\begin{pmatrix} M_2 \mid U_2   \end{pmatrix}\in \mathcal{C}_2, \quad 
Y = \mathrm{rs}\begin{pmatrix} 0&\hat{B}\\ \hat{A}&0 \end{pmatrix}\in \mathcal{C}_4.
\] 
Then
\[
\begin{split}
	d_S(V, Y)
	&= 2\mathrm{rank}\left( 
	\begin{array}{@{}cc@{}}
		U_1 & M_1 \\
		0 & \hat{B} \\
		\hat{A} & 0 \\
	\end{array}
	\right)-2k\\
	&\geq 2\mathrm{rank}\left( 
	\begin{array}{@{}cc@{}}
		U_1 & M_1 \\
		0 & \hat{B}  \\
	\end{array}
	\right)-2k.\\
	&= 2(k+k_2)-2k\\
	&=2k_2\\
	&\geq d.
\end{split}
\]
Similarly, it can be proved that \(d_S(X, Y)\geq d\), so \( \mathcal{C}_1 \cup \mathcal{C}_2 \cup \mathcal{C}_4\) is an \((n, d, k)_q\)-CDC. 

Next, Theorem~\ref{thm:concatenated_cdc} directly yields \( \mathcal{C}_3 \cup \mathcal{C}_4\) is an \((n, d, k)_q\)-CDC. So, \( \mathcal{C}_1 \cup \mathcal{C}_2 \cup \mathcal{C}_3 \cup \mathcal{C}_4\) is an \((n, d, k)_q\)-CDC. 
\end{proof}

The new class of $ (18,8,9)_q $-CDCs by Theorem~\ref{thm:concatenated_cdc2} is constructed as follows.

\begin{example}\label{ex:ccq_construction}  
In Theorem~\ref{thm:concatenated_cdc2}, set \( n = 18 \), \( d = 8 \), \( k = 9 \) with \( n_1 = n_2 = 9 \), \( d_1 = d_2 = 4 \), \( k_1 = 4 \), \( k_2 = 5 \). Then \( \mathcal{C}_1\) is a \((18, M_1, 8, 9)_q\)-CDC where \( M_1 = q^{54} \) and \( \mathcal{C}_2\) is a \((18, M_2, 8, 9)_q\)-CDC where \( M_2 \geq \sum_{i=4}^{5} a(q, 9, 9, 4, i) + 1 \). The explicit constructions of \( \mathcal{C}_3 \) and \( \mathcal{C}_4 \) are as follows.

Firstly, $\mathcal{C}_3$ is constructed as follows. Let $v_1 = (111100000)$ and a list of $(9, q^5, 8, 4)_q$-CDCs $(\mathscr{C}_1, \mathscr{C}_2, \cdots, \mathscr{C}_{q^{10}})$ with distance 4 is obtained by Lemma~\ref{lemma:cpc_sequence}. Set $(\mathcal{A}_i)_{1 \leq i \leq q^{10}} = (\mathscr{C}_i)_{1 \leq i \leq q^{10}}$. Let $X_1 = \{v_2 = (111110000), v_3 = (000011111)\}$ be a $(9,8,5)_2$ CWC. By Lemma~\ref{lemma:cdc_fdrmc}, there exists a list of  $(9,8,5)_q$-CDCs $(\mathcal{B}_i)_{1 \leq i \leq q^{10}}$ with distance 4, shown in Table~\ref{table:1}. Furthermore, these CDCs are reordered and illustrated in Table~\ref{table:2}.  

\begin{table}[H]
  \centering
  \caption{A list of $(9,8,5)_q$-CDCs with distance 2}\label{table:1}
  \vspace{-10pt} 
  \begin{tabular}{|c|c c|c c|}
    \hline
    $(9,8,5)_2$ CWC & & & \multicolumn{2}{c|}{$(9,8,5)_q$ CDCs $(\mathcal{B}_i)$} \\
    \hline
     & $D_{v_i}$ & $s_{v_i}$ & size & number \\
    \hline
    $X_1 \quad v_2 = (111110000)$ & $q^5$ & $q^{10}$ & $q^{5}+1$ & $1$\\
    $\quad \quad  v_3 = (000011111)$ & $1$ & $1$ & $q^5$ & $q^{10} - 1 $ \\
    \hline
  \end{tabular}
\end{table}
\begin{table}[H]
  \centering
  \caption{}\label{table:2}
  \vspace{-10pt}
  \begin{tabular}{|c|c|c|}
    \hline
    $|\mathcal{A}_i|$ & $q^5$ & $q^5$ \\
    \hline
    $|\mathcal{B}_i|$ & $q^5+1$ & $q^5$ \\
    \hline
    number & $1$ & $q^{10}-1$ \\
    \hline
  \end{tabular}
\end{table}
\noindent
The number of $\mathcal{A}_i$ and $\mathcal{B}_i$ is $q^{10}$, then \( s_1 = q^{10} \) and \( \mathcal{C}_3 = \bigcup_{i=1}^{s_1} \mathcal{C}_3^i \) has cardinality
\[
|\mathcal{C}_3| = \sum_{i=1}^{s_1} |\mathcal{C}_3^i| = \sum_{i=1}^{s_1} |\mathcal{A}_i| \cdot |\mathcal{B}_i| = q^{20} + q^5.
\]

Secondly, $\mathcal{C}_4$ is constructed as follows. Let $\hat{v}_1 = (000001111)$ and there exists a CDC $\hat{\mathcal{A}}_1$ with cardinality $1$. Let $X_2 = \{\hat{v}_2 = (000011111), \hat{v}_3=(111110000)\}$ is a $(9,8,5)_2$ CWC. By Remark~\ref{remark}, we obtain an inverse list of CDCs $(\hat{\mathcal{B}}_i)_{1 \leq i \leq q^{10}}$, as detailed in Table~\ref{table:3}. Furthermore, these CDCs are reordered and illustrated in Table~\ref{table:4}.  
\begin{table}[H]
  \centering
  \caption{An inverse list of $(9,8,5)_q$-CDCs with distance 2}\label{table:3}
\vspace{-10pt}
  \begin{tabular}{|c|c c|c c|}
    \hline
    $(9,8,5)_2$ CWC & & & \multicolumn{2}{c|}{$(9,8,5)_q$-CDCs $(\hat{\mathcal{B}}_i)$} \\
    \hline
     & $D_{v_i}$ & $s_{v_i}$ & size & number \\
    \hline
    $X_1 \quad \hat{v}_2 = (111110000)$ & $q^5$ & $q^{10}$ & $q^{5}+1$ & $1$\\
    $\quad \quad \hat{v}_3 = (000011111)$ & $1$ & $1$ & $q^5$ & $q^{10} - 1 $ \\
    \hline
  \end{tabular}
\end{table}
\begin{table}[H]
  \centering
  \caption{}\label{table:4}
\vspace{-10pt}
  \begin{tabular}{|c|c|}
    \hline
    $|\hat{\mathcal{A}}_i|$ & $1$  \\
    \hline
    $|\hat{\mathcal{B}}_i|$ & $q^5+1$  \\
    \hline
    number & $1$  \\
    \hline
  \end{tabular}
\end{table}
\noindent
Then \(
|\mathcal{C}_4| = |\hat{\mathcal{A}}_1| \cdot |\hat{\mathcal{B}}_1| = q^{5} + 1.\)

Finally, set \(\mathcal{C} = \mathcal{C}_1 \cup \mathcal{C}_2 \cup \mathcal{C}_3 \cup \mathcal{C}_4\) and the following result is obtained.
\[
A_q(18,8,9)\geq q^{54}+\left(\sum_{i=4}^5a(q,9,9,4,i)+1 \right)+q^{20}+2q^5+1.
\]
When $q = 2$, \(A_2(18,8,9)\geq 18015215399116937,\) which improves the lower bound $18015215398101558$ given in \cite{25}. For $q \geq 2$, this bound is better than the lower bound of $A_q(18,8,9)$ in \cite{25}.
\end{example}

\section{New CDCs by multilevel inserting constructions}\label{sec:cdc_multilevel}
In this section, we first propose a novel set of identifying vectors and utilize the multilevel construction to establish a new class of CDCs. Next, we insert the coset construction into the multilevel construction and provide three new classes of large CDCs, one of which is constructed by using new optimal FDRMCs.

The multilevel construction is effective for improving the lower bound on the cardinality of CDCs. However, the selection of identifying vectors to ensure the optimality of CDCs remains a critical challenge. 
The new construction of CDCs in Theorem~\ref{th:1} presented below reduces the complexity of selection for identifying vectors.

\begin{theo}\label{th:1} 
Let \( n \geq 2k \), \( k \geq 2\delta + \left\lfloor \frac{\delta}{2} \right\rfloor - 1 \).
\begin{enumerate}[(1)]
    \item If \( \delta = 3 \), then
    \[
\begin{split}
    A_q(n, 2\delta, k) \geq & q^{(n-k)(k-\delta+1)} + q^{\left(n-k-\left\lceil \frac{\delta}{2} \right\rceil\right)(k-\delta+1) - \left\lfloor \frac{\delta}{2} \right\rfloor\left(\delta + \left\lfloor \frac{\delta}{2} \right\rfloor\right)} \\
    &+ \sum_{j=0}^{3} q^{(n-k)(k-\delta+1)- \delta^2 - j{\left\lceil \frac{\delta}{2} \right\rceil}^2 - j{\left\lfloor \frac{\delta}{2} \right\rfloor}^2};
\end{split}
    \]
    
    \item If \( \delta \neq 3 \), then
    \[
    A_q(n, 2\delta, k) \geq q^{(n-k)(k-\delta+1)} + \sum_{j=0}^{2} q^{(n-k)(k-\delta+1)-\delta^2-j( \frac{\delta}{2})^2}.
    \]
\end{enumerate}
\end{theo}
 
\begin{proof}
Let 
\[
\begin{split}
  \begin{aligned}
    &S_1 =
    \{(\underbrace{1\cdots1}_{k}\underbrace{0\cdots0}_{n-k})\},\\
&S_2 =
    \{(\underbrace{0\cdots0}_{i\left\lceil \frac{\delta}{2} \right\rceil}\underbrace{1\cdots1}_{k-\delta}\underbrace{0\cdots0}_{\delta-i\left\lceil \frac{\delta}{2} \right\rceil}\underbrace{1\cdots1}_{\left\lceil \frac{\delta}{2} \right\rceil}\underbrace{0\ldots0}_{i\left\lfloor \frac{\delta}{2} \right\rfloor}\underbrace{1\cdots1}_{\left\lfloor \frac{\delta}{2} \right\rfloor}\underbrace{0\cdots0}_{n-k-\delta-i\left\lfloor \frac{\delta}{2} \right\rfloor})\},\\
&S_3 = \{(\underbrace{1\cdots1}_{k-\delta-\left\lfloor \frac{\delta}{2} \right\rfloor}\underbrace{0\cdots0}_{j\left\lfloor \frac{\delta}{2} \right\rfloor}\underbrace{1\cdots1}_{\left\lfloor \frac{\delta}{2} \right\rfloor}\underbrace{0\ldots0}_{\delta-j\left\lfloor \frac{\delta}{2} \right\rfloor}\underbrace{1\cdots1}_{\left\lfloor \frac{\delta}{2} \right\rfloor}\underbrace{0\cdots0}_{j\left\lceil \frac{\delta}{2} \right\rceil}\underbrace{1\cdots1}_{\left\lceil \frac{\delta}{2} \right\rceil}\underbrace{0\cdots0}_{n-k-\delta-j\left\lceil \frac{\delta}{2} \right\rceil})\}.
  \end{aligned}
\end{split}
\]
Let \( i \in \{0, 1\} \), \( j \in \{1, 2, 3\} \) when \(\delta = 3\) and \( i =0 \), \( j \in \{1, 2\} \) when \(\delta \neq 3\).
Set \(S=S_1 \cup S_2 \cup S_3\). Clearly, \( S \) is a binary vector set of length \( n \) and weight \( k \). 
For any \( v \in S \), the $v$ can be denoted as \(
v = ( \underbrace{v^{(1)}}_{k}|\underbrace{v^{(2)}}_{n-k}).\) It is first proved that $S_2$ and $S_3$ are $(n,2\delta,k)_2$ CWCs. For any $v,u \in S_2$, \( d_H(v^{(1)}, u^{(1)}) \geq 2\left\lceil \frac{\delta}{2} \right\rceil\) and \( d_H(v^{(2)}, u^{(2)}) \geq 2\left\lfloor \frac{\delta}{2} \right\rfloor \), then 
\[d_H(v, u) = d_H(v^{(1)}, u^{(1)}) + d_H(v^{(2)}, u^{(2)})\geq 2( \left\lceil \tfrac{\delta}{2} \right\rceil + \left\lfloor \tfrac{\delta}{2} \right\rfloor ) = 2\delta.\] 
Therefore, \(S_2\) is an $(n,2\delta,k)_2$ CWC. Similarly, \(S_3\) is also an $(n,2\delta,k)_2$ CWC. 

Next, it will be shown that \(S\) is an $(n,2\delta,k)_2$ CWC. For any $v_1 \in S_1, v_2 \in S_2, v_3 \in S_3$, \( d_H(v_1^{(1)}, v_i^{(1)}) = k-(k-\delta)=\delta\) and \( d_H(v_1^{(2)}, v_i^{(2)}) = \delta \) for \(i=2,3\), then 
\[d_H(v_1, v_i) = d_H(v_1^{(1)}, v_i^{(1)}) + d_H(v_1^{(2)}, v_i^{(2)})= \delta+\delta= 2\delta\]
for $i=2,3$. Next, we prove that $d_H(v_2,v_3) \geq 2\delta. $ When $\delta \neq 3$, $v_2=(\underbrace{1\cdots1}_{k-\delta}\underbrace{0\cdots0}_{\delta}\underbrace{1\cdots1}_{\delta}\underbrace{0\cdots0}_{n-k-\delta}),$ then \( d_H(v_2^{(1)}, v_3^{(1)})\geq 2\left\lfloor \frac{\delta}{2} \right\rfloor$ and $d_H(v_2^{(2)}, v_3^{(2)})\geq 2\left\lceil \frac{\delta}{2} \right\rceil$. Therefore, 
\[d_H(v_2, v_3) = d_H(v_2^{(1)}, v_3^{(1)}) + d_H(v_2^{(2)}, v_3^{(2)})\geq 2( \left\lceil \tfrac{\delta}{2} \right\rceil + \left\lfloor \tfrac{\delta}{2} \right\rfloor ) = 2\delta.\] 
Thus, it has been proved that the set $S$ is an $(n,2\delta,k)_2$ CWC when $\delta \neq 3$. The situation of \(v_2 = (00\underbrace{1\cdots1}_{k-3}01101\underbrace{0\cdots0}_{n-k-4})\) is further considered when $\delta = 3$. By \(k - \delta - \left\lfloor \frac{\delta}{2} \right\rfloor \geq \delta - \left\lfloor \frac{\delta}{2} \right\rfloor = \left\lceil \frac{\delta}{2} \right\rceil,\) then \(d_H(v_2^{(1)}, v_3^{(1)})\geq 2\left\lceil \frac{\delta}{2} \right\rceil$ and $d_H(v_2^{(2)}, v_3^{(2)})\geq 2\left\lceil \frac{\delta}{2} \right\rceil-\left\lfloor \frac{\delta}{2} \right\rfloor$. Thus 
\[d_H(v_2, v_3) = d_H(v_2^{(1)}, v_3^{(1)}) + d_H(v_2^{(2)}, v_3^{(2)})\geq 2( 2\left\lceil \tfrac{\delta}{2} \right\rceil -\left\lfloor \tfrac{\delta}{2} \right\rfloor ) = 2\delta,\] 
where \(\delta=3\). In conclusion, $S$ is an $(n, 2\delta, k)_2$ CWC.

If \(\delta=3\), thus \(|S|=6\). Without loss of generality, let \(S=\{v_1, v_2, \cdots, v_6\}\), where 
\[
\begin{aligned}
v_1 &=(\underbrace{1\cdots1}_{k}\underbrace{0\cdots0}_{n-k}) \in S_1 ,\\
v_2 &=(\underbrace{1\cdots1}_{k-3}000111\underbrace{0\cdots0}_{n-k-3}) \in S_2,\\ 
v_3 &=(00\underbrace{1\cdots1}_{k-3}01101\underbrace{0\cdots0}_{n-k-4})\in S_2,\\
v_4 &=(\underbrace{1\cdots1}_{k-4}
      0
      1
      00
      1
      00
      11
      \underbrace{0\cdots0}_{n-k-5}) \in S_3,\\
v_5 &=(\underbrace{1\cdots1}_{k-4}
      00
      1
      0
      1
      0000
      11
      \underbrace{0\cdots0}_{n-k-7})\in S_3,\\
v_6 &=(\underbrace{1\cdots1}_{k-4}
      000
      11
      000000
      11
      \underbrace{0\cdots0}_{n-k-9}) \in S_3.\\
\end{aligned}
\]
Then
\begin{enumerate}[(1)]
    \item For the vector \(v_1\), its corresponding Ferrers diagram is $\mathcal{F}_{v_1}=[\underbrace{k, \cdots, k}_{n-k}].$ Due to \(n-k \geq k\), we consider $\mathcal{F}_{v_1}^t$. Since the rightmost $k-\delta$ columns of \(\mathcal{F}_{v_1}^t\) are full, by Lemma~\ref{lem:op_FDRMC}, there exists an optimal $[\mathcal{F}_{v_1}^t, (n-k)(k-2), 3]_q$ code. By applying Lemma~\ref{lemma:multilevel}, an $(n,q^{(n-k)(k-2)},6,k)_q$-CDC $\mathcal{C}_{v_1}$ is constructed.
    \item For the vector \(v_2\), its corresponding Ferrers diagram is \[\mathcal{F}_{v_2}=[\underbrace{k-3, \cdots, k-3}_{3}, \underbrace{k, \cdots, k}_{n-k-3}].\] Due to \(n-k \geq k\), we consider $\mathcal{F}_{v_2}^t$. Since the rightmost $k-3$ columns of \(\mathcal{F}_{v_2}^t\) are full, by Lemma~\ref{lem:op_FDRMC}, there exists an optimal $[\mathcal{F}_{v_2}^t, (n-k)(k-2)-9, 3]_q$ code. By applying Lemma~\ref{lemma:multilevel}, an $(n, q^{(n-k)(k-2)-9}, 2\delta, k)_q$-CDC $\mathcal{C}_{v_2}$ is constructed.
    \item For the vector $v_3$, its corresponding Ferrers diagram is 
\[
\mathcal{F}_{v_3}=[
k-3,\ 
k -1,\ 
\underbrace{k, \cdots, k}_{n-k-4}
].
\]
Due to \(n-k-2 \geq k\), we consider $\mathcal{F}_{v_3}^t$. Since the rightmost $k-3$ columns of \(\mathcal{F}_{v_3}^t\) are full and $k-3 \geq 2,$ by Lemma~\ref{lem:op_FDRMC}, there exists an optimal $[\mathcal{F}_{v_3}^t, (n-k-2)(k-2)-4, 3]_q$ code. By applying Lemma~\ref{lemma:multilevel}, an $(n, q^{(n-k-2)(k-2)-4}, 6, k)_q$-CDC $\mathcal{C}_{v_3}$ is constructed.
    \item For the vector $v_4$, its corresponding Ferrers diagram is 
\[
\mathcal{F}_{v_4}=[k-4, k-3, k-3, k-2, k-2,
\underbrace{k, \cdots, k}_{n-k-5}]. 
\]
Due to \(n-k \geq k\), we consider $\mathcal{F}_{v_4}^t$. Since the rightmost $k-4$ columns of \(\mathcal{F}_{v_4}^t\) are full and $k-4 \geq 2$, by Lemma~\ref{lem:op_FDRMC}, there exists an optimal $[\mathcal{F}_{v_4}^t, (n-k)(k-2)-14, 3]_q$ code. By applying Lemma~\ref{lemma:multilevel}, an $(n, q^{(n-k)(k-2)-14}, 6, k)_q$-CDC $\mathcal{C}_{v_4}$ is constructed.
    \item For the vector $v_5$, its corresponding Ferrers diagram is 
\[
\mathcal{F}_{v_5}=[k-4, k-4, k-3,
\underbrace{k-2, \cdots, k-2}_{4},
\underbrace{k, \cdots, k}_{n-k-7}].
\]
 Due to \(n-k \geq k\), we consider $\mathcal{F}_{v_5}^t$. Since the rightmost $k-4$ columns of \(\mathcal{F}_{v_5}^t\) are full and $k-4 \geq 2$, by Lemma~\ref{lem:op_FDRMC}, there exists an optimal $[\mathcal{F}_{v_5}^t, (n-k)(k-2)-19, 3]_q$ code. By applying Lemma~\ref{lemma:multilevel}, an $(n, q^{(n-k)(k-2)-19}, 6, k)_q$-CDC $\mathcal{C}_{v_5}$ is constructed.
    \item For the vector $v_6$, its corresponding Ferrers diagram is 
\[
\mathcal{F}_{v_6}=[k-4, k-4, k-4, 
\underbrace{k-2, \cdots, k-2}_{6},
\underbrace{k, \cdots, k}_{n-k-9}]. 
\]
Due to \(n-k \geq k\), we consider $\mathcal{F}_{v_6}^t$. Since the rightmost $k-4$ columns of \(\mathcal{F}_{v_6}^t\) are full and $k-4 \geq 2$, by Lemma~\ref{lem:op_FDRMC}, there exists an optimal $[\mathcal{F}_{v_6}^t, (n-k)(k-2)-24, 3]_q$ code. By applying Lemma~\ref{lemma:multilevel}, an $(n, q^{(n-k)(k-2)-24}, 6, k)_q$-CDC $\mathcal{C}_{v_6}$ is constructed.
\end{enumerate}

Finally, set $ \mathcal{C}_5 = \bigcup_{i=1}^6 \mathcal{C}_{v_i}$. For any $\mathcal{U} \in \mathcal{C}_{v_i}$ and $\mathcal{V} \in \mathcal{C}_{v_j}$, where $v_i,v_j \in S$. Since $S$ is an $(n,2\delta,k)_2$ CWC, $\mathcal{C}_{v_i}$ and $\mathcal{C}_{v_j} $ are lifted $ [\mathcal{F}_{v_i},\delta]_q$ and $[\mathcal{F}_{v_j},\delta]_q$ code, respectively. By Lemma~\ref{lemma:subspace_hamming}, $d_S(\mathcal{U}, \mathcal{V}) \geq 2\delta$. Thus, $ \mathcal{C}_5 $ is an $(n, N_1, 2\delta, k)_q$-CDC, where
\[
\begin{split}
    N_1 = & q^{(n-k)(k-\delta+1)} + q^{\left(n-k-\left\lceil \frac{\delta}{2} \right\rceil\right)(k-\delta+1) - \left\lfloor \frac{\delta}{2} \right\rfloor\left(\delta + \left\lfloor \frac{\delta}{2} \right\rfloor\right)} \\
    &+ \sum_{j=0}^{3} q^{(n-k)(k-\delta+1)- \delta^2 - j{\left\lceil \frac{\delta}{2} \right\rceil}^2 - j{\left\lfloor \frac{\delta}{2} \right\rfloor}^2}.
\end{split}
\]

Furthermore, if \( \delta \neq 3 \), thus \(|S|=4\). Without loss of generality, let \(S=\{v_1, v_2, v_4, v_5\}\), where
\[
\begin{aligned}
	v_1 &=(\underbrace{1\cdots1}_{k}\underbrace{0\cdots0}_{n-k}) \in S_1 ,\\
	v_2 &=(\underbrace{1\cdots1}_{k-\delta}\underbrace{0\cdots0}_{\delta}\underbrace{1\cdots1}_{\delta}\underbrace{0\cdots0}_{n-k-\delta})\in S_2,\\
	v_4 &=(\underbrace{1\cdots1}_{k-\delta-\left\lfloor \frac{\delta}{2} \right\rfloor}
	\underbrace{0\cdots0}_{\left\lfloor \frac{\delta}{2} \right\rfloor}
	\underbrace{1\cdots1}_{\left\lfloor \frac{\delta}{2} \right\rfloor}
	\underbrace{0\cdots0}_{\left\lceil \frac{\delta}{2} \right\rceil}
	\underbrace{1\cdots1}_{\left\lfloor \frac{\delta}{2} \right\rfloor}
	\underbrace{0\cdots0}_{\left\lceil \frac{\delta}{2} \right\rceil}
	\underbrace{1\cdots1}_{\left\lceil \frac{\delta}{2} \right\rceil}
	\underbrace{0\cdots0}_{n-k-\delta-\left\lceil \frac{\delta}{2} \right\rceil}) \in S_3,\\
	v_5 &=(\underbrace{1\cdots1}_{k-\delta-\left\lfloor \frac{\delta}{2} \right\rfloor}
	\underbrace{0\cdots0}_{2\left\lfloor \frac{\delta}{2} \right\rfloor}
	\underbrace{1\cdots1}_{\left\lfloor \frac{\delta}{2} \right\rfloor}
	\underbrace{0\cdots0}_{\left\lceil \frac{\delta}{2} \right\rceil-\left\lfloor \frac{\delta}{2} \right\rfloor}
	\underbrace{1\cdots1}_{\left\lfloor \frac{\delta}{2} \right\rfloor}
	\underbrace{0\cdots0}_{2\left\lceil \frac{\delta}{2} \right\rceil}
	\underbrace{1\cdots1}_{\left\lceil \frac{\delta}{2} \right\rceil}
	\underbrace{0\cdots0}_{n-k-\delta-2\left\lceil \frac{\delta}{2} \right\rceil}) \in S_3.\\
\end{aligned}
\]
Set $ \mathcal{C}_5 = \mathcal{C}_{v_1} \cup \mathcal{C}_{v_2} \cup \mathcal{C}_{v_4} \cup \mathcal{C}_{v_5}$. Similar to the case of \( \delta = 3 \), it can be proved that $ \mathcal{C}_5 $ is an $(n, N_1, 2\delta, k)_q$-CDC, where
    \[
    N_1 = q^{(n-k)(k-\delta+1)} + \sum_{j=0}^{2} q^{(n-k)(k-\delta+1)-\delta^2-j( \frac{\delta}{2})^2}.
    \]
Therefore, the new lower bounds of $A_q(n, 2\delta, k)$ for CDCs are obtained.
\end{proof}

Based on Theorem~\ref{th:1}, the following example is obtained.

\begin{example}\label{ex:4.3}
    Let \( n = 19 \), \( \delta = 4 \) and \( k = 9 \) in Theorem~\ref{th:1}. Then
    \[
    A_q(19,8,9) \geq q^{60} + q^{44} + q^{36} + q^{28}.
    \]
When $q = 3$, we have $A_3(19,8,9) \geq 42391159260137223209995120164$, which improves the lower bound $42391159259987125499483652096$ given in \cite{table}. For $q \geq 3$, this bound is better than the lower bound of $A_q (19, 8, 9)$ in \cite{table}.
\end{example}

To achieve larger CDCs, we insert the coset construction of \cite{co1} into Theorem~\ref{th:1} under the condition of parameter constraints.

\begin{theo}\label{th:2}
Keep the notations as in Lemma~\ref{lemma:coset_construction}, where $n_1 \geq k + 1$, $n_2 \geq k_2$, $k_1 \geq \delta$ and $k_2 \geq \delta$. Define
\[
	\mathcal{C}_6^i = \left\{ \mathrm{rs}\left(
	\begin{array}{@{}cc@{}}
		A_i & \varphi_{B_i}(H) \\
		0_{k_2 \times n_1} & B_i \\
	\end{array}
	\right)\,\bigg|\,  
	\tau^{-1}(A_i) \in \mathcal{A}_i, \, \tau^{-1}(B_i) \in \mathcal{B}_i, \, H \in \mathcal{H}  
	\right\}.
\]  
Set $\mathcal{C}_6 = \bigcup_{1 \leq i \leq s} \mathcal{C}_6^i$. Furthermore, if $|k - \delta + 1 - k_1| \geq \delta$, then $\mathcal{C}_5 \cup \mathcal{C}_6$ is an $(n, 2\delta, k)_q$-CDC and 
\begin{enumerate}[(1)]
	\item if \( \delta = 3 \), then
	\[
	\begin{split}
		A_q(n, 2\delta, k) \geq & q^{(n-k)(k-\delta+1)} + q^{\left(n-k-\left\lceil \frac{\delta}{2} \right\rceil\right)(k-\delta+1) - \left\lfloor \frac{\delta}{2} \right\rfloor\left(\delta + \left\lfloor \frac{\delta}{2} \right\rfloor\right)}\\
		&+ \sum_{j=0}^{3} q^{(n-k)(k-\delta+1)- \delta^2 - j{\left\lceil \frac{\delta}{2} \right\rceil}^2 - j{\left\lfloor \frac{\delta}{2} \right\rfloor}^2}+|\mathcal{H}| \cdot \sum_{i=1}^s|\mathcal{A}_i| \cdot |\mathcal{B}_i|;
	\end{split}
	\]
	\item if \( \delta \neq 3 \), then
	\[
	A_q(n, 2\delta, k) \geq q^{(n-k)(k-\delta+1)} + \sum_{j=0}^{2} q^{(n-k)(k-\delta+1)-\delta^2-j( \frac{\delta}{2})^2}+|\mathcal{H}| \cdot \sum_{i=1}^s|\mathcal{A}_i| \cdot |\mathcal{B}_i|.
	\]
\end{enumerate}
\end{theo}

\begin{proof}
By Lemma~\ref{lemma:coset_construction}, the $\mathcal{C}_6$ is an $(n, 2\delta, k)_q$-CDC. Below is the proof that $\mathcal{C}_5 \cup \mathcal{C}_6$ is an $(n, 2\delta, k)_q$-CDC. For any $\mathcal{U} \in \mathcal{C}_5$ and $\mathcal{V} \in \mathcal{C}_6$, let $u$ and $v$ be the identifying vectors of $\mathcal{U}$ and $\mathcal{V}$, respectively. In the first $n_1$ positions, $u$ contains at least $k - \delta + 1$ ones and $v$ contains $k_1$ ones. By Lemma~\ref{lemma:multi_coset}, suppose that $|(k - \delta + 1) - k_1| \geq \delta$, then $d_S({\mathcal{U}, \mathcal{V}}) \geq d_H(u, v) \geq 2\delta$.
\end{proof}

The new class of  $ (17,6,8)_q $-CDCs by Theorem~\ref{th:2} is constructed as follows.

\begin{example}
Let \( n = 17 \), \( \delta = 3 \), \( k = 8 \), \( n_1 =9, n_2 = 8 \), \( \delta_1 = 1, \delta_2 = 2 \), \( k_1 = 3 \), \( k_2 = 5 \) in Theorem~\ref{th:2}. Then \( \mathcal{C}_5\) is a \((17, N_1, 6, 9)_q\)-CDC where \( N_1 = q^{54} + q^{45} +q^{40}+q^{38}+q^{35}+q^{30} \) by Theorem~\ref{th:1}. The explicit construction of \( \mathcal{C}_6 \) is as follows.

Let $\mathcal{H}$ be a $[9 \times 3, 3]_q$ MRD code, then $|\mathcal{H}|=q^9$. By Lemma~\ref{lemma:cascaded_codes}, we obtain a list of $(9, 6, 3)_q$-CDCs $(\mathcal{A}_i)_{1 \leq i \leq q^{12}}$ with distance 2 (Table~\ref{table:5}) and a list of $(8,6,5)_q$-CDCs $(\mathcal{B}_i)_{1 \leq i \leq q^5+2q^3+q}$ with distance 4 (Table~\ref{table:6}). Furthermore, these CDCs are reordered and illustrated in Table~\ref{table:7}.  

\begin{table}[H]
  \centering
  \caption{A list of $(9,6,3)_q$-CDCs with distance 2}\label{table:5}
\vspace{-10pt}
  \begin{tabular}{|c|c c|c c|}
    \hline
    $(9,6,3)_2$ CWC & & & \multicolumn{2}{c|}{$(9,6,3)_q$-CDCs $(\mathcal{A}_i)$} \\
    \hline
     & $D_{v_i}$ & $s_{v_i}$ & size & number \\
    \hline
    $U_1 \quad  (111000000)$ & $q^6$ & $q^{12}$ & $q^{6}+q^3+1$ & $1$\\
    $\quad \quad (000111000)$ & $q^3$ & $q^6$ & $q^6+q^3$ & $q^{6} - 1 $ \\
    $\quad \quad (000000111)$ & $1$ & $1$ & $q^6$ & $q^{12} - q^6 $ \\
    \hline
  \end{tabular}
\end{table}
\begin{table}[H]
  \centering
  \caption{A list of $(8,6,5)_q$-CDCs with distance 4}\label{table:6}
\vspace{-10pt}
  \begin{tabular}{|c|c c|c c|c c|}
    \hline
    $(8,6,5)_2$ CWC & & & \multicolumn{2}{c|}{$(8,6,5)_q$-CDCs} &  \multicolumn{2}{c|}{reordered $(8,6,5)_q$-CDCs $(\mathcal{B}_i)$} \\
    \hline
     & $D_{v_i}$ & $s_{v_i}$ & size & number & size & number\\
    \hline
    $V_1 \quad (11111000)$ & $q^5$ & $q^{5}$ & $q^{5}+1$ & $q^3 $ & $q^5+1$ & $q^{3}$\\
    $\quad \quad (11000111)$ & $1$ & $q^{3}$ & $q^{5} $ & $ q^5-q^3 $ & $q^5 $ & $q^5-q^{3}$\\
    \cline{1-5}
    $V_2 \quad (10110110)$ & $q $ & $q^{3}$ & $q +1$ & $q^2 $ & $q +1$ & $2q^{2}$\\
    $\quad \quad (01101101)$ & $1$ & $q^{2}$ & $q  $ & $ q^3-q^2 $ & $q  $ & $q^3-q^{2}$\\
    \cline{1-5}
    $V_3 \quad (01110101)$ & $ 1 $ & $q^{3}$ & $q +1$ & $q^2 $ & $1$ & $q^{3}-q^2+q$\\
    $\quad \quad (10101011)$ & $q$ & $q^{2}$ & $1  $ & $ q^3-q^2 $ &  & \\
    \cline{1-5}
    $V_4 \quad (01011011)$ & $ 1 $ & $q $ & $1$ & $q $ &   &  \\
    \hline
  \end{tabular}
\end{table}
\begin{table}[ht]
  \centering
  \caption{}\label{table:7}
\vspace{-10pt}
  \begin{tabular}{|c|c|c|c|c|c|c|}
    \hline
    $|\mathcal{A}_i|$ & $q^6+q^3+1$ & $q^5+q^3$ & $q^5+q^3$ & $q^5+q^3$ & $q^5+q^3$ & $q^5+q^3$ \\
    \hline
    $|\mathcal{B}_i|$ & $q^5+1$ & $q^5+1$ & $q^5$ & $q+1$ & $q$ & $1$ \\
    \hline
    number & $1$ & $q^{3}-1$ & $q^5-q^3$ & $2q^2$ & $q^3-q^2$ & $q^3-q^2+q$\\
    \hline
  \end{tabular}
\end{table}

The number of $\mathcal{A}_i$ is $q^{12}$ and the number of $\mathcal{B}_i$ is $q^5+2q^3+q$, then \( s = \min\{q^{12},q^5+2q^3+q\}\). To obtain large CDCs, by applying Lemma~\ref{lemma:rearrangement}, \( \mathcal{C}_6 = \bigcup_{i=1}^{s} \mathcal{C}_6^i \) has cardinality
\[
\begin{split}
|\mathcal{C}_6| &= |\mathcal{H}| \cdot \sum_{i=1}^{s} |\mathcal{C}_6^i|= |\mathcal{H}| \cdot \sum_{i=1}^{s} |\mathcal{A}_i| \cdot |\mathcal{B}_i| \\
&= q^{25} + q^{22} +q^{19} + 3q^{18} + q^{17} + 2q^{16} + 3q^{15} + 2q^{14} + q^{13} + q^{12}.
\end{split}
\]

Finally, $\mathcal{C}_5 \cup \mathcal{C}_6$ is a $(17, 6, 8)_q$-CDC and
\[
\begin{split}
A_q(17,6,8)\geq & q^{54}+q^{45}+q^{40}+q^{38}+q^{35}+q^{30}+q^{25} + q^{22} +q^{19} + 3q^{18} + q^{17}\\
&+ 2q^{16} + 3q^{15} + 2q^{14} + q^{13} + q^{12}.
\end{split}
\]
When $q = 3$, $A_3(17,6,8) \geq 58152704874502104749268072$, which improves the lower bound $58151863451946414791142287$ given in \cite{11}. For $q \geq 3$, this bound is better than the lower bound of $A_q (17,6,8)$ in \cite{11}.
\end{example}

Subsequently, a new class of optimal FDRMC is constructed based on Lemma~\ref{lem:composite_FDRMC}.

\begin{theo}\label{th:3}
Let \( n\) and \( k \) be positive integers with \( n \geq 2k \geq 4 \), $\delta=3$. Given a Ferrers diagram 
$$
\tikzset{every picture/.style={line width=0.75pt}} 
\begin{tikzpicture}[x=0.65pt,y=0.62pt,yscale=-1,xscale=1]

\draw   (377,84.5) .. controls (376.96,79.83) and (374.61,77.52) .. (369.94,77.56) -- (357.44,77.67) .. controls (350.77,77.72) and (347.42,75.42) .. (347.38,70.75) .. controls (347.42,75.42) and (344.11,77.78) .. (337.44,77.84)(340.44,77.81) -- (324.94,77.94) .. controls (320.27,77.98) and (317.96,80.33) .. (318,85) ;
\draw   (438,167.5) .. controls (442.67,167.59) and (445.05,165.31) .. (445.14,160.64) -- (445.3,152.89) .. controls (445.43,146.22) and (447.83,142.94) .. (452.5,143.03) .. controls (447.83,142.94) and (445.57,139.56) .. (445.7,132.89)(445.64,135.89) -- (445.86,125.14) .. controls (445.95,120.47) and (443.67,118.09) .. (439,118) ;
\draw   (441,240.5) .. controls (445.67,240.59) and (448.05,238.31) .. (448.14,233.64) -- (448.3,225.89) .. controls (448.43,219.22) and (450.83,215.94) .. (455.5,216.03) .. controls (450.83,215.94) and (448.57,212.56) .. (448.7,205.89)(448.64,208.89) -- (448.86,198.14) .. controls (448.95,193.47) and (446.67,191.09) .. (442,191) ;

\draw (304,80.4) node [anchor=north west][inner sep=0.75pt]    {$\begin{array}{ c c c c c }
\bullet  & \cdots  & \bullet  & \bullet  & \bullet \\
 &  &  & \bullet  & \bullet \\
 &  &  & \vdots  & \vdots \\
 &  &  & \bullet  & \bullet \\
 &  &  &  & \bullet \\
 &  &  &  & \vdots \\
 &  &  &  & \bullet 
\end{array}$};
\draw (310,48) node [anchor=north west][inner sep=0.75pt]   [align=left] {$n-k-2$};
\draw (460,119.4) node [anchor=north west][inner sep=0.75pt]    {$\lfloor \frac{k-1}{2} \rfloor $};
\draw (467,192.4) node [anchor=north west][inner sep=0.75pt]    {$\lfloor \frac{k-1}{2} \rfloor $};
\draw (264,144.4) node [anchor=north west][inner sep=0.75pt]    {$\mathcal{F} =$};
\end{tikzpicture},
$$
there exists an optimal \([\mathcal{F}, \left\lfloor \frac{k - 1}{2} \right\rfloor, 3]_q\) code.
\end{theo}

\begin{proof}
Let 

$$
\tikzset{every picture/.style={line width=0.75pt}} 
\begin{tikzpicture}[x=0.65pt,y=0.65pt,yscale=-1,xscale=1]

\draw   (377,84.5) .. controls (376.96,79.83) and (374.61,77.52) .. (369.94,77.56) -- (357.44,77.67) .. controls (350.77,77.72) and (347.42,75.42) .. (347.38,70.75) .. controls (347.42,75.42) and (344.11,77.78) .. (337.44,77.84)(340.44,77.81) -- (324.94,77.94) .. controls (320.27,77.98) and (317.96,80.33) .. (318,85) ;
\draw   (413,165.5) .. controls (417.67,165.59) and (420.05,163.31) .. (420.14,158.64) -- (420.3,150.89) .. controls (420.43,144.22) and (422.83,140.94) .. (427.5,141.03) .. controls (422.83,140.94) and (420.57,137.56) .. (420.7,130.89)(420.64,133.89) -- (420.86,123.14) .. controls (420.95,118.47) and (418.67,116.09) .. (414,116) ;

\draw (307,80.4) node [anchor=north west][inner sep=0.75pt]    {$\begin{array}{ c c c c }
\bullet  & \cdots  & \bullet  & \bullet \\
 &  &  & \bullet \\
 &  &  & \vdots \\
 &  &  & \bullet \\
\end{array}$};
\draw (314,50) node [anchor=north west][inner sep=0.75pt]   [align=left] {$n-k-1$};
\draw (436,125) node [anchor=north west][inner sep=0.75pt]    {$\lfloor \frac{k-1}{2} \rfloor $};
\draw (263,119.4) node [anchor=north west][inner sep=0.75pt]    {$\mathcal{F}_{1} =$};
\end{tikzpicture}.
$$
Since $ n - k - 1 \geq k - 1 $ and by Lemma~\ref{lem:op_FDRMC}, there exists an optimal $[\mathcal{F}_1,\left\lfloor \frac{k - 1}{2} \right\rfloor,2]_q$ code. Let 

$$
\tikzset{every picture/.style={line width=0.75pt}} 
\begin{tikzpicture}[x=0.65pt,y=0.65pt,yscale=-1,xscale=1]

\draw   (102,82.5) .. controls (106.67,82.59) and (109.05,80.31) .. (109.14,75.64) -- (109.3,67.89) .. controls (109.43,61.22) and (111.83,57.94) .. (116.5,58.03) .. controls (111.83,57.94) and (109.57,54.56) .. (109.7,47.89)(109.64,50.89) -- (109.86,40.14) .. controls (109.95,35.47) and (107.67,33.09) .. (103,33) ;

\draw (76,22.4) node [anchor=north west][inner sep=0.75pt]    {$\begin{array}{ c }
\bullet \\
\vdots \\
\bullet 
\end{array}$};
\draw (125,43) node [anchor=north west][inner sep=0.75pt]    {$\lfloor \frac{k-1}{2} \rfloor $};
\draw (34,51.4) node [anchor=north west][inner sep=0.75pt]    {$\mathcal{F}_{2} =$};
\end{tikzpicture},
$$
and there exists an optimal $[\mathcal{F}_2,\left\lfloor \frac{k - 1}{2} \right\rfloor,1]_q$ code by Lemma~\ref{lem:op_FDRMC}. Furthermore, there exists an $[\mathcal{F},\left\lfloor \frac{k - 1}{2} \right\rfloor,3]_q$ code by Lemma~\ref{lem:composite_FDRMC}. It is clear that  $v_{\text{min}}(\mathcal{F},3)=v_1(\mathcal{F},3)=\left\lfloor \frac{k - 1}{2} \right\rfloor$, so there exists an optimal $[\mathcal{F},\left\lfloor \frac{k - 1}{2} \right\rfloor,3]_q$ code.
\end{proof}

In the following, the new class of optimal FDRMC from Theorem~\ref{th:3} is applied to $\mathcal{C}_5$ of Theorem~\ref{th:1} to improve the lower bounds on the number of codewords for more CDCs.

\begin{theo}\label{th:4}
Let \( n \geq 2k+2 \), \( k \geq 2\delta + \left\lfloor \frac{\delta}{2} \right\rfloor - 1 \).
\begin{enumerate}[(1)]
    \item If \( \delta = 3 \), then
    \[
\begin{split}
    A_q(n, 2\delta, k) \geq & q^{(n-k)(k-\delta+1)} + q^{\left(n-k-\left\lceil \frac{\delta}{2} \right\rceil\right)(k-\delta+1) - \left\lfloor \frac{\delta}{2} \right\rfloor\left(\delta + \left\lfloor \frac{\delta}{2} \right\rfloor\right)} \\
    &+ \sum_{j=0}^{3} q^{(n-k)(k-\delta+1)- \delta^2 - j{\left\lceil \frac{\delta}{2} \right\rceil}^2 - j{\left\lfloor \frac{\delta}{2} \right\rfloor}^2}+q^{\left\lfloor \frac{k - 1}{2} \right\rfloor};
\end{split}
    \]
    
    \item If \( \delta \neq 3 \), then
    \[
    A_q(n, 2\delta, k) \geq q^{(n-k)(k-\delta+1)} + \sum_{j=0}^{2} q^{(n-k)(k-\delta+1)-\delta^2-j( \frac{\delta}{2})^2}+q^{\left\lfloor \frac{k - 1}{2} \right\rfloor}.
    \]
\end{enumerate}
\end{theo}

\begin{proof}
Let \( S_1, S_2, S_3, \mathcal{C}_5 \) be defined in Theorem~\ref{th:1}. Given the identifying vector 
\[
v' = (1\underbrace{0\cdots0}_{n-k-2}\underbrace{1\cdots1}_{\left\lfloor \frac{k-1}{2} \right\rfloor}0\underbrace{1\cdots1}_{\left\lfloor \frac{k-1}{2} \right\rfloor}0\underbrace{1\cdots1}_{\left\lceil \frac{k-1}{2} \right\rceil-\left\lfloor \frac{k-1}{2} \right\rfloor})
\]
and set \( S' =S_1 \cup S_2 \cup S_3 \cup v'\). We now prove that \( S' \) is an \((n, 2\delta, k)_2 \) CWC. For any \( v \in S' \) of the form $ v = ( \underbrace{v^{(3)}}_{k+1}|\underbrace{v^{(4)}}_{n-k-1})$, clearly, the first \( k+1 \) positions of \( v' \) contain l one, while \( v_1, v_2, v_3 \) contain at least \( k-\delta+1 \) ones for any \( v_1 \in S_1 \), \( v_2 \in S_2 \) and \( v_3 \in S_3 \). Then 
\[
d_H ({v'}^{(3)}, v_i^{(3)} ) \geq k-\delta+1 - 1 = k-\delta \geq \delta \] 
for  $i \in \{1, 2, 3\}$. In the remaining \( n-k-1 \) positions, \( v' \) contains \( k-1 \) ones and \( v_1, v_2, v_3 \) contain at least \( \delta-1 \) ones, then $d_H ({v'}^{(4)}, v_i^{(4)} )  \geq (k-1) - (\delta - 1) = k - \delta \geq \delta $ for $ i \in \{1, 2, 3\}$. Hence, 
\[
d_H (v', v_i) = d_H ({v'}^{(3)}, v_i^{(3)}) + d_H ({v'}^{(4)}, v_i^{(4)}) \geq 2\delta
\]
for $ i \in \{1, 2, 3\}$. Regarding $v'$, the corresponding Ferrers diagram coincides exactly with the Ferrers diagram of Theorem~\ref{th:3}. Thus, there exists an optimal $[\mathcal{F},\left\lfloor \frac{k - 1}{2} \right\rfloor,3]_q$ code $\mathcal{C}_{\mathcal{F}_{v'}}$ by Theorem~\ref{th:3}. Set 
\[
\mathcal{C}_5 \cup \mathcal{C}_{\mathcal{F}_{v'}} \triangleq \mathcal{C}_7.
\]
If \( \delta = 3 \), then $ \mathcal{C}_7 $ is an $(n, N_2, 2\delta, k)_q$-CDC, where
\[
\begin{split}
    N_2 = & q^{(n-k)(k-\delta+1)} + q^{\left(n-k-\left\lceil \frac{\delta}{2} \right\rceil\right)(k-\delta+1) - \left\lfloor \frac{\delta}{2} \right\rfloor\left(\delta + \left\lfloor \frac{\delta}{2} \right\rfloor\right)} \\
    &+ \sum_{j=0}^{3} q^{(n-k)(k-\delta+1)- \delta^2 - j{\left\lceil \frac{\delta}{2} \right\rceil}^2 - j{\left\lfloor \frac{\delta}{2} \right\rfloor}^2}+q^{\left\lfloor \frac{k - 1}{2} \right\rfloor}.
\end{split}
\]
Otherwise, $ \mathcal{C}_7 $ is an $(n, N_2, 2\delta, k)_q$-CDC, where
    \[N_2 = q^{(n-k)(k-\delta+1)} + \sum_{j=0}^{2} q^{(n-k)(k-\delta+1)-\delta^2-j( \frac{\delta}{2})^2}+q^{\left\lfloor \frac{k - 1}{2} \right\rfloor}.\]
\end{proof}

Based on Theorem~\ref{th:4}, the following results are obtained. Through comparing with known bounds, our constructions yield tighter lower bounds of CDCs.

\begin{example}

With $n=15, \delta=3, k=6$, then
\[
A_q(15,6,6) \geq q^{36}+q^{27}+q^{24}+2q^{22}+q^{12}+q^2.
\]
When $q = 3$, then we have $A_3(15,6,6) \geq 150102606086671257$, which improves the lower bound $150102543990846750$ given in \cite{6}. For $q \geq 3$, this bound is better than the lower bound of $A_q (15,6,6)$ in \cite{6}.

With $n=16, \delta=3, k=6$, then
\[
A_q(16,6,6) \geq q^{40}+q^{31}+q^{28}+q^{26}+q^{21}+q^{16}+q^2.
\]
When $q = 3$, then we have $A_3(16,6,6) \geq 12158308561614895971$, which improves the lower bound $12158306011247867322$ given in \cite{6}. For $q \geq 3$, this bound is better than the lower bound of $A_q (16,6,6)$ in \cite{6}.

With $n=16, \delta=3, k=7$, then
\[
A_q(16,6,7) \geq q^{45}+q^{36}+2q^{31}+q^{26}+q^{21}+q^3.
\]
When $q = 3$, then we have $A_3(16,6,7) \geq 2954464039085249447217$, which improves the lower bound $2954462808812115984384$ given in \cite{table}. For $q \geq 3$, this bound is better than the lower bound of $A_q (16,6,7)$ in \cite{table}.

With $n=17, \delta=3, k=6$, then
\[
A_q(17,6,6) \geq q^{44}+q^{35}+q^{32}+q^{30}+q^{25}+q^{20}+q^2.
\]
When $q = 3$, then we have $A_3(17, 6, 6) \geq 984822993490806572931$, which improves the lower bound $984822786755034714042$ given in \cite{26}. For $q \geq 3$, this bound is better than the lower bound of $A_q (17, 6, 6)$ in \cite{26}.

With $n=17, \delta=3, k=7$, then
\[
A_q(17,6,7) \geq q^{50}+q^{41}+2q^{36}+q^{31}+q^{26}+q^3.
\]
When $q = 3$, then we have $A_3(17,6,7) \geq 717934761497715615667197$, which improves the lower bound $717934513945701079533438$ given in \cite{26}. For $q \geq 3$, this bound is better than the lower bound of $A_q (17,6,7)$ in \cite{26}.

With $n=18, \delta=3, k=7$, then
\[
A_q(18,6,7) \geq q^{55}+q^{46}+2q^{41}+q^{36}+q^{31}+q^3.
\]
When $q = 3$, then we have $A_3(18,6,7) \geq 174458147043944894607122337$, which improves the lower bound $174458086133951448364668984$ given in \cite{26}. For $q \geq 3$, this bound is better than the lower bound of $A_q (18,6,7)$ in \cite{26}.

With $n=19, \delta=3, k=7$, then
\[
A_q(19,6,7) \geq q^{60}+q^{51}+2q^{46}+q^{41}+q^{36}+q^3.
\]
When $q = 3$, then we have $A_3(19,6,7) \geq 42393329731678609389530721357$, which improves the lower bound $42393314923753645026362778948$ given in \cite{26}. For $q \geq 3$, this bound is better than the lower bound of $A_q (19,6,7)$ in \cite{26}.
\end{example}

To improve the lower bounds of CDCs with additional parameters, the coset construction in \cite{co1} is inserted into the CDCs of Theorem~\ref{th:4} under parameter constraints.

\begin{theo}\label{th:5}
Keep the notations as in Lemma~\ref{lemma:coset_construction}, where $n \geq 2k+2, \delta \geq 2, n_1 = k+1, n_2 \geq k_2$, $k_1 \geq \delta$ and $k_2 \geq \delta$. Define
    \[
    \mathcal{C}_8^i = \left\{ \mathrm{rs}\left(
    \begin{array}{@{}cc@{}}
    A_i & \varphi_{B_i}(H) \\
    0_{k_2 \times n_1} & B_i \\
    \end{array}
    \right)\,\bigg|\,  
    \tau^{-1}(A_i) \in \mathcal{A}_i, \, \tau^{-1}(B_i) \in \mathcal{B}_i, \, H \in \mathcal{H}  
    \right\}.
    \]  
Set \( \mathcal{C}_8=\bigcup_{1 \leq i \leq s} \mathcal{C}_8^i \) and for each $\mathcal{U} \in \mathcal{C}_8$, let $v(\mathcal{U})$ be the identifying vector of  $\mathcal{U}$. If the parameters satisfy the following conditions that
\begin{enumerate}[(1)]
    \item $|k - \delta + 1 - k_1| \geq \delta$;
    \item $v(\mathcal{U})$ has the form $(0, a_1, a_2, \cdots, a_{n-\delta-1}, \overbrace{0, \cdots, 0}^{L_1}, 1, \overbrace{0, \ldots, 0}^{ L_2})$, where $\sum_{i=1}^{n-\delta-1} a_i = k - 1$, $a_i \in \{0,1\}$ for $1 \leq i \leq n-\delta-1$, $L_1=\delta-1 - \left\lceil \frac{k-1}{2} \right\rceil + \left\lfloor \frac{k-1}{2} \right\rfloor$ and $L_2=\left\lceil \frac{k-1}{2} \right\rceil - \left\lfloor \frac{k-1}{2} \right\rfloor$,
\end{enumerate}
then $\mathcal{C}_7 \cup \mathcal{C}_8$ is an \((n, 2\delta, k)_q\)-CDC, where $\mathcal{C}_7$ is defined in Theorem~\ref{th:4}.  
\end{theo}

\begin{proof}
By Lemma~\ref{lemma:coset_construction}, $\mathcal{C}_8$ is an $(n, 2\delta, k)_q$-CDC. Below we prove that \(\mathcal{C}_7 \cup \mathcal{C}_8\) is an \( (n, 2\delta, k)_q \)-CDC, where \( \mathcal{C}_7 = \mathcal{C}_5 \cup \mathcal{C}_{\mathcal{F}_{v'}} \) is defined in Theorem~\ref{th:4}. It is first proved that \( \mathcal{C}_5 \cup \mathcal{C}_8 \) is an \( (n, 2\delta, k)_q \)-CDC. For any $\mathcal{V} \in \mathcal{C}_5$ and $\mathcal{U} \in \mathcal{C}_8$, let $v(\mathcal{V})$ and $v(\mathcal{U})$ be the identifying vectors of $\mathcal{\mathcal{V}}$ and $\mathcal{U}$, respectively, where $v(\mathcal{U})$ satisfies conditions (1) and (2). Considering the first \( n_1 \) positions, $v(\mathcal{V})$ has at least $k-\delta+1$ ones and $v(\mathcal{U})$ has $k$ ones, then \( \mathcal{C}_5 \cup \mathcal{C}_8 \) is an \( (n, 2\delta, k)_q \)-CDC by condition (1) and Lemma~\ref{lemma:multi_coset}. Next we prove that \( \mathcal{C}_{\mathcal{F}_{v'}} \cup \mathcal{C}_8 \) also is an \( (n, 2\delta, k)_q \)-CDC, where $v' = (1\underbrace{0\cdots0}_{n-k-2}\underbrace{1\cdots1}_{\left\lfloor \frac{k-1}{2} \right\rfloor}0\underbrace{1\cdots1}_{\left\lfloor \frac{k-1}{2} \right\rfloor}0\underbrace{1\cdots1}_{\left\lceil \frac{k-1}{2} \right\rceil-\left\lfloor \frac{k-1}{2} \right\rfloor})$. Due to $ k \geq \delta$, then 
\[
2\left\lfloor \tfrac{k-1}{2} \right\rfloor \geq \delta-1-\left\lceil \tfrac{k-1}{2} \right\rceil + \left\lfloor \tfrac{k-1}{2} \right\rfloor.
\]
 Furthermore, by considering the first position and last $\delta-1$ positions, $v(\mathcal{U})$ and $v'$ differ in at least $\delta$ components. Moreover, in the remaining positions, they differ in at least $\delta$ components. Therefore, $d_H (v(\mathcal{U}), v') \geq 2\delta$.
\end{proof}

The following corollary is obtained by Theorem~\ref{th:5}.

\begin{coro}\label{cor:5}
Keep the notations as in Theorem~\ref{th:5}.
\begin{enumerate}[(1)]
    \item If \( \delta = 3 \), then
    \[
\begin{split}
    A_q(n, 2\delta, k) \geq & q^{(n-k)(k-\delta+1)} + q^{\left(n-k-\left\lceil \frac{\delta}{2} \right\rceil\right)(k-\delta+1) - \left\lfloor \frac{\delta}{2} \right\rfloor\left(\delta + \left\lfloor \frac{\delta}{2} \right\rfloor\right)} \\
    &+ \sum_{j=0}^{3} q^{(n-k)(k-\delta+1)- \delta^2 - j{\left\lceil \frac{\delta}{2} \right\rceil}^2 - j{\left\lfloor \frac{\delta}{2} \right\rfloor}^2}+q^{\left\lfloor \frac{k - 1}{2} \right\rfloor}+|\mathcal{H}| \cdot \sum_{i=1}^s|\mathcal{A}_i| \cdot |\mathcal{B}_i|;
\end{split}
    \]
    \item If \( \delta \neq 3 \), then
    \[
\begin{split}
    A_q(n, 2\delta, k) \geq &q^{(n-k)(k-\delta+1)} + \sum_{j=0}^{2} q^{(n-k)(k-\delta+1)-\delta^2-j( \frac{\delta}{2})^2}+q^{\left\lfloor \frac{k - 1}{2} \right\rfloor}\\
&+|\mathcal{H}| \cdot \sum_{i=1}^s|\mathcal{A}_i| \cdot |\mathcal{B}_i|.
\end{split}
    \]
\end{enumerate}
\end{coro}

The new class of $ (19,6,8)_q $-CDCs by Theorem~\ref{th:5} is constructed as follows.
\begin{example}
Let \( n = 19 \), \( \delta = 3 \), \( k = 8 \) with \( n_1 =9, n_2 = 10 \), \( \delta_1 = 1, \delta_2 = 2 \), \( k_1 = 3 \), \( k_2 = 5 \) in Theorem~\ref{th:5}. Then \( \mathcal{C}_7\) is a \((19, N_2, 6, 8)_q\)-CDC where \( N_2 = q^{66} + q^{57} +q^{52}+q^{50}+q^{47}+q^{42}+q^{3} \). The explicit construction of \( \mathcal{C}_8 \) is as follows.

Let $\mathcal{H}$ be a $[3 \times 5, 3]_q$ MRD code, then $|\mathcal{H}|=q^5$. By Lemma~\ref{lemma:cascaded_codes}, we obtain a list of $(9, 6, 3)_q$-CDCs $(\mathcal{A}_i)_{1 \leq i \leq q^{10}}$ with distance 2 (Table~\ref{table:8}) and a list of $(10,6,5)_q$-CDCs $(\mathcal{B}_i)_{1 \leq i \leq q^{5}+3q^{4}}$ with distance 4 (Table~\ref{table:9}). Furthermore, these CDCs are reordered and illustrated in Table~\ref{table:10}.  

\begin{table}[H]
  \centering
  \caption{A list of $(9,6,3)_q$-CDCs with distance 2}\label{table:8}
\vspace{-10pt}
  \begin{tabular}{|c|c c|c c|}
    \hline
    $(9,6,3)_2$ CWC & & & \multicolumn{2}{c|}{$(9,6,3)_q$-CDCs $(\mathcal{A}_i)$} \\
    \hline
     & $D_{v_i}$ & $s_{v_i}$ & size & number \\
    \hline
    $U_1 \quad  (011100000)$ & $q^5$ & $q^{10}$ & $q^{5}+1$ & $q^6$\\
    $\quad \quad (000011100)$ & $1$ & $q^6$ & $q^5$ & $q^{10} - q^6 $ \\
    \hline
  \end{tabular}
\end{table}

\begin{table}[H]
  \centering
  \caption{A list of $(8,6,5)_q$-CDCs with distance 4}\label{table:9}
\vspace{-10pt}
  \begin{tabular}{|c|c c|c c|c c|}
    \hline
    $(10,6,5)_2$ CWC & & & \multicolumn{2}{c|}{$(10,6,5)_q$-CDCs} &  \multicolumn{2}{c|}{reordered $(10,6,5)_q$-CDCs $(\mathcal{B}_i)$} \\
    \hline
     & $D_{v_i}$ & $s_{v_i}$ & size & number & size & number\\
    \hline
    $V_1 \quad (1111000010)$ & $q^{11}$ & $q^{5}$ & $q^{11}+q^3$ & $q^4 $ & $q^{11}+q^3$ & $q^{4}$\\
    $\quad \quad (1000111010)$ & $q^{3}$ & $q^{4}$ & $ q^{11}  $ & $q^5-q^4$ & $ q^{11}$ & $q^5-q^4$\\
    \cline{1-5}
    $V_2 \quad (0110110010)$ & $q^6 $ & $q^{4}$ & $q^6 $ & $q^{4}$ & $q^6 $ & $q^{4}$\\
    \cline{1-5}
    $V_3 \quad (0101101010)$ & $q^6 $ & $q^{4}$ & $q^6 $ & $q^{4}$ & $q^6 $ & $q^{4}$\\
    \cline{1-5}
    $V_4 \quad (0011011010)$ & $ q^2 $ & $q^{4}$ & $ q^2 $ & $q^{4}$ & $ q^2 $ & $q^{4}$\\
    \hline
  \end{tabular}
\end{table}

\begin{table}[H]
  \centering
  \caption{}\label{table:10}
\vspace{-10pt}
  \begin{tabular}{|c|c|c|c|c|c|c|}
    \hline
    &$|\mathcal{A}_i|$ & $q^5+1$ & $q^5+1$ & $q^5+1$ & $q^5+1$ & $q^5$ \\
    \cline{2-7}
    $q = 2$ & $|\mathcal{B}_i|$ & $q^{11}+q^3$ & $q^{11}$ & $q^6$ & $q^4$ & $q^2$ \\
    \cline{2-7}
    &number & $q^4$ & $q^{5}-q^4$ & $q^4$ & $q^4$ & $q^4$ \\
    \hline
    &$|\mathcal{A}_i|$ & $q^5+1$ & $q^5+1$ & $q^5+1$ & $q^5+1$ & $q^5+1$ \\
    \cline{2-7}
    $q \geq 3$ & $|\mathcal{B}_i|$ & $q^{11}+q^3$ & $q^{11}$ & $q^6$ & $q^4$ & $q^2$ \\
    \cline{2-7}
    &number & $q^4$ & $q^{5}-q^4$ & $q^4$ & $q^4$ & $q^4$ \\
    \hline
  \end{tabular}
\end{table}

The number of $\mathcal{A}_i$ is $q^{10}$ and the number of $\mathcal{B}_i$ is $q^{5}+3q^4$, then \( s = \min\{q^{10},q^{5}+3q^4\}\). To obtain large CDCs, by applying Lemma~\ref{lemma:rearrangement}, \( \mathcal{C}_8 = \bigcup_{i=1}^{s} \mathcal{C}_8^i \) has cardinality
\[
\begin{split}
|\mathcal{C}_8| &= |\mathcal{H}| \cdot \sum_{i=1}^{s} |\mathcal{C}_8^i|= |\mathcal{H}| \cdot \sum_{i=1}^{s} |\mathcal{A}_i| \cdot |\mathcal{B}_i| \\
&=
\begin{cases} 
\displaystyle q^{26}+q^{21}+q^{20}+q^{18}+q^{17}+q^{16}+q^{15}+q^{13}+q^{12}, & \text{if } q = 2; \\ 
\displaystyle q^{26}+q^{21}+q^{20}+q^{18}+q^{17}+q^{16}+q^{15}+q^{12}+q^{11}, & \text{if } q \geq 3.
\end{cases}
\end{split}
\]

Finally, $\mathcal{C}_7 \cup \mathcal{C}_8$ is a $(19, 6, 8)_q$-CDC and
\[
\begin{split}
A_q(19,6,8)
& \geq |\mathcal{C}_7|+|\mathcal{C}_8|\\
& = q^{66} + q^{57} +q^{52}+q^{50}+q^{47}+q^{42}+q^{26}+q^{21}+q^{20}+q^{18}+q^{17}\\
& + q^{16}+q^{15}+q^{12}+q^{11}+q^{3},
\end{split}
\]
where $q \geq 3$. When $q = 3$, we have $A_3(19,6,8) \geq 30904731631209804712703574912729$, which improves the lower bound $30904724644711433087390807634522$ given in \cite{26}. For $q \geq 3$, this bound is better than the lower bound of $A_q (19,6,8)$ in \cite{26}.
\end{example}

\section{Conclusions}

The paper aims to construct new classes of CDCs to improve the lower bounds of $A_q(n,d,k)$. The main contributions of this paper are summarized as follows. Firstly, by cosets of optimal FDRMCs, the construction of parallel cosets of optimal FDRMCs is proposed by using the list of CDCs and inverse list of CDCs. Then a new class of $ (18,8,9)_q $-CDCs is obtained by combining the parallel cosets of optimal FDRMCs with parallel linkage construction. Next, we present a new set of identifying vectors and establish new lower bounds of $ (19,8,9)_q $-CDCs in Theorem~\ref{th:1} through the multilevel construction which reduces the complexity of selection for identifying vectors. By inserting the coset construction into Theorem~\ref{th:1}, a new construction of CDCs is provided to improve the lower bounds on the cardinality of some CDCs. Furthermore, a new class of optimal FDRMCs is constructed and the new optimal FDRMCs are applied to Theorem~\ref{th:1}, thereby further increasing the number of codewords for more CDCs. Finally, under the specific parameter constraints, the coset construction is inserted into Theorem~\ref{th:4}, yielding a new class of large CDCs.

For the constructions in this paper, the lower bounds for CDCs with minimum distances $d = 6$ and $d = 8$ are established, which is larger than the previously
best-known results (see Table~\ref{table:11}). The choosing of better set of identifying vectors and verification of optimality for FDRMCs are critical to the constructions in this paper. If more effective sets of identifying vectors and additional new classes of optimal FDRMCs can be provided, further improvements to the lower bounds of $A_q(n,d,k)$ for more CDCs could be achieved.

\section*{Acknowledgement}
G. Wang and H. Yao are supported by the National Natural Science Foundation of China (No. 12301670), the Natural Science Foundation of Tianjin (No. 23JCQNJC00050), the Fundamental Research Funds for the Central Universities of China (No. 3122024PT24) and the Graduate Student Research and Innovation Fund of Civil Aviation University of China (No. 2024YJSKC06001). F.-W Fu is supported by the National Key Research and Development Program of China (Grant No. 2018YFA0704703), the National Natural Science Foundation of China (Grant No. 61971243), the Natural Science Foundation of Tianjin (20JCZDJC00610), the Fundamental Research Funds for the Central Universities of China (Nankai University) and the Nankai Zhide Foundation. 
\appendix 

\begin{table}[H]
\centering
\caption{New lower bounds of $A_q(n,2\delta,k)$}\label{table:11}
\vskip 2mm 
\setlength{\tabcolsep}{10pt}
\vspace{-10pt}
{
\begin{tabular}{|c|c|c|c|>{\centering\arraybackslash}m{1cm}|}

\cline{1-5}
$A_q(n,2\delta,k)$&  New lower bounds &Old lower bounds&Differences & Refs. \\
\cline{1-5}
$A_2(18,8,9)$&\makecell[c]{1801521539\underline{91}\\\underline{16937}}&\makecell[c]{1801521539\underline{81}\\\underline{01558}}&\makecell[c]{1015379}&\\
\cline{1-4}
$A_3(18,8,9)$&\makecell[c]{5814973938041\\76\underline{70685308834}}&\makecell[c]{5814973938041\\76\underline{67212878011}}&\makecell[c]{3472430823}&\\
\cline{1-4}
$A_4(18,6,9)$&\makecell[c]{324518553767\\84298642\underline{4312}\\\underline{397880897}}&\makecell[c]{324518553767\\84298642\underline{3213}\\\underline{960023182}}&\makecell[c]{1098437857715}& \\
\cline{1-4}
$A_5(18,8,9)$&\makecell[c]{555111512317\\35878357960\underline{2}\\\underline{12196595695126}}&\makecell[c]{555111512317\\35878357960\underline{1}\\\underline{16859681752031}}&\makecell[c]{95336913943095}& \cite{25}\\
\cline{1-4}
$A_7(18,8,9)$&\makecell[c]{431811456739\\659181762301\\6095\underline{36509153}\\\underline{0833971362}}& \makecell[c]{431811456739\\659181762301\\6095\underline{28530401}\\\underline{2098975554}}&\makecell[c]{797875187349\\95808}&\\
\cline{1-4}
$A_8(18,8,9)$&\makecell[c]{584600654932\\363583793403\\430292\underline{508651}\\\underline{1686985425409}}&\makecell[c]{584600654932\\363583793403\\430292\underline{393362}\\\underline{5366753377384}}&\makecell[c]{115288632023\\2048025}&\\
\cline{1-4}
$A_9(18,8,9)$&\makecell[c]{338139191352\\272842462028\\02470185\underline{2687}\\\underline{107871571291}\\\underline{3542}}&\makecell[c]{338139191352\\272842462028\\02470185\underline{1471}\\\underline{361914779407}\\\underline{9924}}&\makecell[c]{121574595679\\18833618}&\\
\cline{1-5}
$A_3(19,8,9)$&\makecell[c]{423911592\underline{601}\\\underline{372232099951}\\\underline{20164}}& \makecell[c]{423911592\underline{599}\\\underline{871254994836}\\\underline{52096}}&\makecell[c]{150097710511\\468068}&\\
\cline{1-4}
$A_4(19,8,9)$&\makecell[c]{132922799609\\440\underline{560516369}\\\underline{2592688267264}}&\makecell[c]{132922799609\\440\underline{088272515}\\\underline{2129005125632}}&\makecell[c]{472243854046\\3683141632}&\\
\cline{1-4}
$A_7(19,8,9)$&\makecell[c]{508021860739\\6386\underline{52025471}\\\underline{871534256456}\\\underline{974722786020}\\\underline{804}}& \makecell[c]{508021860739\\6386\underline{38374792}\\\underline{724238783200}\\\underline{605413176246}\\\underline{272}}&\makecell[c]{136506791472\\954732563693\\09609774532}&\cite{table}\\
\cline{1-4}
$A_8(19,8,9)$&\makecell[c]{153249554086\\5894302876\underline{54}\\\underline{228073872583}\\\underline{745307525329}\\\underline{0860544} }&\makecell[c]{153249554086\\5894302876\underline{21}\\\underline{776216572459}\\\underline{761245803047}\\\underline{4985472}}&\makecell[c]{324518573001\\239840617222\\815875072}&\\
\hline
$A_3(17,6,8)$&\makecell[c]{5815\underline{270487450}\\\underline{2104749268072}}& \makecell[c]{5815\underline{186345194}\\\underline{6414791142287}}&\makecell[c]{841422555689\\958125785}&\\
\cline{1-4}
$A_4(17,6,8)$&\makecell[c]{324519\underline{792884}\\\underline{131441190199}\\\underline{176986624}}& \makecell[c]{324519\underline{094951}\\\underline{964764830545}\\\underline{503899935}}&\makecell[c]{697932166676\\359653673086\\689}&\\
\cline{1-4}
$A_5(17,6,8)$&\makecell[c]{555111\underline{796624}\\\underline{2891363829615}\\\underline{9010253906250}}& \makecell[c]{555111\underline{600407}\\\underline{3007983442483}\\\underline{7423236913732}}&\makecell[c]{196216988338\\038713215870\\16992518}&\\
\hline
\end{tabular}}
\end{table}

\begin{table}[H]
	\centering
	
	\vskip 2mm \setlength{\tabcolsep}{10pt}
{
\begin{tabular}{|c|c|c|c|>{\centering\arraybackslash}m{1cm}|}
\hline
$A_7(17,6,8)$&\makecell[c]{4318114\underline{67440}\\\underline{983810704691}\\\underline{555602040067}\\\underline{5343538188}\\}&  \makecell[c]{4318114\underline{58814}\\\underline{229328190145}\\\underline{779776047452}\\\underline{2447137650}\\}&\makecell[c]{8626754482\\514545775\\825992615\\2896400538}& \cite{11}\\
\cline{1-4}
$A_8(17,6,8)$&	\makecell[c]{58460065\underline{9288}\\\underline{110467644413}\\\underline{889455604591}\\\underline{4458116063232}}& \makecell[c]{58460065\underline{5642}\\\underline{087187407545}\\\underline{566975906539}\\\underline{0165175356426}}&\makecell[c] {3646023280\\2368683224\\7969805242\\92940706806
}&\\
\cline{1-4}
$A_9(17,6,8)$&\makecell[c]{33813919\underline{222}\\\underline{50839547150}\\\underline{90928115048}\\\underline{34902356584}\\\underline{30212094}}& \makecell[c]{33813919\underline{147}\\\underline{48407703492}\\\underline{58063849227}\\\underline{15712541982}\\\underline{93516660}}&\makecell[c]{75024318436\\5832864265\\8211918981\\460136695434}&\\
\cline{1-5}
$A_3(15,6,6)$&	\makecell[c]{150102\underline{606086}\\\underline{671257}} &\makecell[c]{150102\underline{543990}\\\underline{846750}}&\makecell[c]{62095824507}&\\
\cline{1-4}
$A_4(15,6,6)$&\makecell[c]{4722384\underline{81392}\\\underline{7520272400}}& \makecell[c]{ 4722384\underline{77884}\\\underline{1908199452}}&\makecell[c]{35085612072948}&\\
\cline{1-4}
$A_5(15,6,6)$	&\makecell[c]{145519227\underline{4332}\\\underline{0465332031275}}&\makecell[c]{ 145519227\underline{3855}\\\underline{7090682988320}}&\makecell[c]{476337464904\\2955}&\\
\cline{1-4}
$A_7(15,6,6)$ & 
\makecell[c]{26517309117\underline{7}\\\underline{141670880487}\\\underline{7198893}} & 
\makecell[c]{26517309117\underline{6}\\\underline{359901081761}\\\underline{6918746}} & 
\makecell[c]{781769798726\\0280147} &\cite{6} \\
\cline{1-4}
$A_8(15,6,6)$&	\makecell[c]{324518556081\\\underline{148306447942}\\\underline{411092032}}&\makecell[c]{324518556081\\\underline{000753445053}\\\underline{205203320}}&\makecell[c]{147553002889\\205888712}&\\
\cline{1-4}
$A_9(15,6,6)$&\makecell[c]{225283996031\\\underline{706473997617}\\\underline{29727028495}}&\makecell[c]{225283996031\\\underline{686780291297}\\\underline{80303636252}}&\makecell[c]{196937063194\\9423392243}&\\
\hline
$A_3(16,6,6)$ & \makecell[c]{1215830\underline{85616}\\\underline{14895971}} & \makecell[c]{1215830\underline{60112}\\\underline{47867322}} & \makecell[c]{2550367028649}& \\
\cline{1-4}
$A_4(16,6,6)$&\makecell[c]{120893050\underline{786}\\\underline{6243608870928}}& \makecell[c]{120893050\underline{335}\\\underline{ 8636753567772}}&\makecell[c]{450760685530\\3156}&\\
\cline{1-4}
$A_5(16,6,6)$&\makecell[c]{909495171\underline{308}\\\underline{565155029296}\\\underline{8775}}& \makecell[c]{909495171\underline{159}\\\underline{508342711658}\\\underline{2070}}&\makecell[c]{149056812317\\6386705}&\\
\cline{1-4}
$A_7(16,6,6)$&\makecell[c]{63668059191\underline{5}\\\underline{378459604872}\\\underline{6483402851}}& \makecell[c]{63668059191\underline{4}\\\underline{439657074085}\\\underline{0492991138}}&\makecell[c]{938802530787\\5990411713}&\cite{26}\\
\cline{1-4}
$A_8(16,6,6)$&\makecell[c]{132922800570\\\underline{808124097924}\\\underline{0495393800256}}& \makecell[c]{132922800570\\\underline{777900046828}\\\underline{1490544937336}}&\makecell[c]{302240510959\\004848862920}&\\
\cline{1-4}
$A_9(16,6,6)$&\makecell[c]{147808829796\\3\underline{96156617366}\\\underline{471197246851}\\\underline{303}}& \makecell[c]{147808829796\\3\underline{89695427596}\\\underline{956712706655}\\\underline{942}}&\makecell[c]{646118976951\\4484540195361}&\\
\cline{1-5}
$A_3(16,6,7)$ & \makecell[c]{295446\underline{403908}\\\underline{5249447217}} & \makecell[c]{295446\underline{280881}\\\underline{2115984384}} & \makecell[c]{123027313346\\2833}& \\
\cline{1-4}
$A_4(16,6,7)$&\makecell[c]{12379447\underline{7087}\\\underline{974317907299}\\\underline{5392}}& \makecell[c]{12379447\underline{6166}\\\underline{987754305381}\\\underline{9904}}&\makecell[c]{920986563601\\9175488}&\\
\cline{1-5}
\end{tabular}}
\end{table}

\begin{table}[H]
	\centering
	
	\vskip 2mm \setlength{\tabcolsep}{10pt}
{
\begin{tabular}{|c|c|c|c|>{\centering\arraybackslash}m{1cm}|}

\cline{1-5}
$A_5(16,6,7)$&\makecell[c]{284217239\underline{916}\\\underline{339521408081}\\\underline{05468875}}& \makecell[c]{284217239\underline{823}\\\underline{266882297455}\\\underline{25522432}}&\makecell[c]{93072639110\\62579946443}& \\
\cline{1-4}
$A_7(16,6,7)$&\makecell[c]{107006907075\\\underline{644439368118}\\\underline{347289028841}\\\underline{893}}& \makecell[c]{107006907075\\\underline{328954409717}\\\underline{156935706345}\\\underline{472}}&\makecell[c]{315484958401\\190353322496\\421}&\cite{table}\\
\cline{1-4}
$A_8(16,6,7)$&\makecell[c]{43556143290\underline{4}\\\underline{184843246079}\\\underline{208945148860}\\\underline{83072}}& \makecell[c]{43556143290\underline{3}\\\underline{986769817386}\\\underline{765344223516}\\\underline{42624}}&\makecell[c]{198073428692\\443600925344\\40448}&\\
\cline{1-4}
$A_9(16,6,7)$&\makecell[c]{872796359061\\6\underline{87501814071}\\\underline{709119593732}\\\underline{5903887}}& \makecell[c]{872796359061\\6\underline{11247377647}\\\underline{646651075755}\\\underline{6649984}}&\makecell[c]{762544364240\\624685179769\\253903}&\\
\cline{1-5}
$A_3(17,6,6)$ & \makecell[c]{984822\underline{993490}\\\underline{806572931}} & \makecell[c]{984822\underline{786755}\\\underline{034714042}} & \makecell[c]{206735771858\\889}& \\
\cline{1-4}
$A_4(17,6,6)$&\makecell[c]{3094862\underline{10013}\\\underline{758363870953}\\\underline{488}}& \makecell[c]{3094862\underline{08859}\\\underline{711457728282}\\\underline{652}}&\makecell[c]{115404690614\\2670836}&\\
\cline{1-4}
$A_5(17,6,6)$&\makecell[c]{56843448\underline{2067}\\\underline{853221893310}\\\underline{5468775}}& \makecell[c]{56843448\underline{1974}\\\underline{691165142485}\\\underline{0957070}}&\makecell[c]{931620567508\\254511705}& \\
\cline{1-4}
$A_7(17,6,6)$&\makecell[c]{152867010118\\\underline{882368151129}\\\underline{922866501276}\\\underline{51}}& \makecell[c]{152867010118\\\underline{656961337693}\\\underline{834729514940}\\\underline{02}}&\makecell[c]{225406813436\\088136986336\\49}& \cite{26}\\
\cline{1-4}
$A_8(17,6,6)$&\makecell[c]{54445179113\underline{8}\\\underline{030076305096}\\\underline{906913300558}\\\underline{6496}}& \makecell[c]{54445179113\underline{7}\\\underline{906278523294}\\\underline{194583219170}\\\underline{7512}}&\makecell[c]{123797781802\\712330081387\\8984}&\\
\cline{1-4}
$A_9(17,6,6)$&\makecell[c]{969773732294\\1\underline{55183566541}\\\underline{417525136590}\\\underline{867623}}& \makecell[c]{969773732294\\1\underline{12791690369}\\\underline{937653123006}\\\underline{406102}}&\makecell[c]{423918761714\\798720135844\\61521}&\\
\cline{1-5}
$A_3(17,6,7)$ & \makecell[c]{717934\underline{76149}\\\underline{77156156671}\\\underline{97}} & \makecell[c]{717934\underline{51394}\\\underline{57010795334}\\\underline{38}} & \makecell[c]{24755201453\\6133759}& \\
\cline{1-4}
$A_4(17,6,7)$&\makecell[c]{12676554\underline{453}\\\underline{80857015370}\\\underline{747215936}}& \makecell[c]{12676554\underline{371}\\\underline{38715070300}\\\underline{646782528}}&\makecell[c]{82421419450\\70100433408}&\\
\cline{1-4}
$A_5(17,6,7)$&\makecell[c]{888178874\underline{73}\\\underline{85610044002}\\\underline{5329589843875}}& \makecell[c]{888178874\underline{47}\\\underline{68854330508}\\\underline{8386195735750}}&\makecell[c]{26167557134\\93694339410\\8125}& \\
\cline{1-4}
$A_7(17,6,7)$&\makecell[c]{17984650872\underline{2}\\\underline{035609245996}\\\underline{506288670773}\\\underline{9931193}}& \makecell[c]{17984650872\underline{1}\\\underline{543261077131}\\\underline{432064842750}\\\underline{7089598}}&\makecell[c]{492348168865\\074223828023\\2841595}& \cite{26}\\
\cline{1-5}
\end{tabular}}
\end{table}

\begin{table}[H]
	\centering
	
	\vskip 2mm \setlength{\tabcolsep}{10pt}
{
\begin{tabular}{|c|c|c|c|>{\centering\arraybackslash}m{1cm}|}

\cline{1-5}
$A_8(17,6,7)$&\makecell[c]{14272477033\underline{4}\\\underline{043289434875}\\\underline{235187146378}\\\underline{7153326592}}& \makecell[c]{14272477033\underline{3}\\\underline{982450431728}\\\underline{689308408589}\\\underline{8498974208}}&\makecell[c]{60839003146\\54587873778\\88654352384}&\\
\cline{1-4}
$A_9(17,6,7)$&\makecell[c]{515377522062\\\underline{335852946191}\\\underline{203518028903}\\\underline{157255577471}}& \makecell[c]{515377522062\\\underline{293302837556}\\\underline{983424271177}\\\underline{106514588358}}&\makecell[c]{425501086342\\200937577260\\50740989113}&\\
\cline{1-5}
$A_3(18,6,7)$ & \makecell[c]{174458\underline{147043}\\\underline{944894607122}\\\underline{337}} & \makecell[c]{174458\underline{086133}\\\underline{951448364668}\\\underline{984}} & \makecell[c]{609099934462\\42453353}& \\
\cline{1-4}
$A_4(18,6,7)$&\makecell[c]{12980791\underline{7606}\\\underline{999758373964}\\\underline{5149052992}}& \makecell[c]{12980791\underline{6760}\\\underline{321574330640}\\\underline{4506875456}}&\makecell[c]{846678184043\\3240642177536}&\\
\cline{1-4}
$A_5(18,6,7)$&\makecell[c]{277555898\underline{355}\\\underline{800313875079}\\\underline{154968261718}\\\underline{875}}& \makecell[c]{277555898\underline{273}\\\underline{931997862851}\\\underline{979819594173}\\\underline{250}}&\makecell[c]{818683160122\\271751486675\\45625}& \\
\cline{1-4}
$A_7(18,6,7)$&\makecell[c]{30226802720\underline{9}\\\underline{125248459746}\\\underline{328119368969}\\\underline{85017796293}}& \makecell[c]{30226802720\underline{8}\\\underline{297537838299}\\\underline{555148538388}\\\underline{66597817940}}&\makecell[c]{827710621446\\772970830581\\18419978353}& \cite{26}\\
\cline{1-4}
$A_8(18,6,7)$&\makecell[c]{467680527430\\\underline{593050820199}\\\underline{170661241253}\\\underline{774401889899}\\\underline{52}}& \makecell[c]{467680527430\\\underline{393663434529}\\\underline{588211578061}\\\underline{956672312202}\\\underline{24}}&\makecell[c]{199387385669\\582449663191\\817729577697\\28}&\\
\cline{1-4}
$A_9(18,6,7)$&\makecell[c]{304325273002\\5\underline{88697806196}\\\underline{443765360887}\\\underline{025327845510}\\\underline{39087}}& \makecell[c]{304325273002\\5\underline{63570083080}\\\underline{541783083282}\\\underline{130947691001}\\\underline{99966}}&\makecell[c]{251277231159\\019822776048\\943801545083\\9121}&\\
\cline{1-5}
$A_3(19,6,7)$ & \makecell[c]{423933\underline{297316}\\\underline{786093895307}\\\underline{21357}} & \makecell[c]{423933\underline{149237}\\\underline{536450263627}\\\underline{78948}} & \makecell[c]{148079249643\\63167942409}& \\
\cline{1-4}
$A_4(19,6,7)$&\makecell[c]{13292330\underline{7629}\\\underline{567752574939}\\\underline{6632630198336}}& \makecell[c]{13292330\underline{6762}\\\underline{526364036239}\\\underline{2736535656000}}&\makecell[c]{867041388538\\700389609454\\2336}&\\
\cline{1-4}
$A_5(19,6,7)$&\makecell[c]{867362182\underline{361}\\\underline{875980859622}\\\underline{359275817871}\\\underline{093875}}& \makecell[c]{867362182\underline{106}\\\underline{035125792941}\\\underline{607845067641}\\\underline{048250}}&\makecell[c]{255840855066\\680751430750\\230045625}& \\
\cline{1-4}
$A_7(19,6,7)$&\makecell[c]{5080218733\underline{30}\\\underline{376805086295}\\\underline{653670223427}\\\underline{627194096531}\\\underline{993}}& \makecell[c]{5080218733\underline{28}\\\underline{985670761663}\\\underline{997418589797}\\\underline{038070776963}\\\underline{944}}&\makecell[c]{139113432463\\165625163363\\058912331956\\8049}& \cite{26}\\
\cline{1-5}
\end{tabular}}
\end{table}

\begin{table}[H]
	\centering
	
	\vskip 2mm \setlength{\tabcolsep}{10pt}
{
\begin{tabular}{|c|c|c|c|>{\centering\arraybackslash}m{1cm}|}

\cline{1-5}
$A_8(19,6,7)$&\makecell[c]{153249555228\\\underline{456730892762}\\\underline{864242275534}\\\underline{036796011280}\\\underline{5970432}}& \makecell[c]{153249555228\\\underline{391395614936}\\\underline{944645131131}\\\underline{639634566138}\\\underline{9971968}}&\makecell[c]{653352778259\\195971444023\\971614451415\\998464}&\\
\cline{1-4}
$A_9(19,6,7)$&\makecell[c]{179701030455\\2\underline{98600167580}\\\underline{938079007950}\\\underline{179585839495}\\\underline{4264002271}}& \makecell[c]{179701030455\\2\underline{83762496487}\\\underline{559260774270}\\\underline{525668143762}\\\underline{3562556750}}&\makecell[c]{148376710933\\788182336796\\539176957330\\701445521}&\\
\cline{1-5}
$A_3(19,6,8)$ & \makecell[c]{309047\underline{316312}\\\underline{098047127035}\\\underline{74912729}} & \makecell[c]{309047\underline{246447}\\\underline{114330873908}\\\underline{07634522}} & \makecell[c]{698649837162\\5312767278207}& \\
\cline{1-4}
$A_4(19,6,8)$&\makecell[c]{54445386\underline{6149}\\\underline{233614204609}\\\underline{209399128568}\\\underline{6336}}& \makecell[c]{54445386\underline{4011}\\\underline{610047167549}\\\underline{954073075190}\\\underline{5792}}&\makecell[c]{213762356703\\705925532605\\33780544}&\\
\cline{1-4}
$A_5(19,6,8)$&\makecell[c]{1355253409\underline{72}\\\underline{726839936512}\\\underline{988075619745}\\\underline{19335937625}}& \makecell[c]{1355253409\underline{49}\\\underline{709248369776}\\\underline{453457075313}\\\underline{42150218750}}&\makecell[c]{230175915667\\365346185444\\3177185718875}& \\
\cline{1-4}
$A_7(19,6,8)$&\makecell[c]{59768265375\underline{3}\\\underline{591744195505}\\\underline{215756103519}\\\underline{686849900095}\\\underline{34834449}}& \makecell[c]{59768265375\underline{2}\\\underline{693264671475}\\\underline{524357890584}\\\underline{532511927153}\\\underline{68746610}}&\makecell[c]{898479524029\\691398212935\\154337972941\\66087839}& \cite{26}\\
\cline{1-4}
$A_8(19,6,8)$&\makecell[c]{401734514057\\9\underline{95696026883}\\\underline{807366764104}\\\underline{231585069398}\\\underline{499558162944}}& \makecell[c]{401734514057\\9\underline{02970381580}\\\underline{059532262902}\\\underline{425878332092}\\\underline{722483822592}}&\makecell[c]{927256453037\\478345012018\\057067373057\\77074340352}&\\
\cline{1-4}

$A_9(19,6,8)$&\makecell[c]{955004953261\\\underline{902203514935}\\\underline{252229977723}\\\underline{629972735205}\\\underline{787694503979}\\\underline{521}}& \makecell[c]{955004953261\\\underline{859955471378}\\\underline{337350051110}\\\underline{369915665833}\\\underline{762290880889}\\\underline{782}}&\makecell[c]{422480435569\\148799266132\\600570693720\\254036230897\\39}&\\
\cline{1-5}
\end{tabular}}
\end{table}

\end{document}